\documentclass[
 reprint,
 nofootinbib,
 amsmath,amssymb,
 aps,
 prb,
]{revtex4-1}

\usepackage{graphicx}
\usepackage[dvipsnames]{xcolor}
\usepackage{transparent}
\usepackage{dcolumn}
\usepackage{bm}	
\usepackage{amssymb}
\usepackage{braket}
\usepackage{hyperref}
\usepackage{mathtools}
\usepackage{setspace}
\usepackage{amsmath}
\usepackage{esdiff}
\usepackage[sort&compress]{natbib}
\usepackage{caption}
\usepackage{subfigure}
\usepackage{amsthm}
\usepackage{fancyhdr}
\usepackage{titlesec}

\usepackage{xr}
\newcommand{\da}{\downarrow}
\newcommand{\ua}{\uparrow}
\newcommand{\Da}{\Downarrow}
\newcommand{\Ua}{\Uparrow}

\DeclareRobustCommand*{\citen}[1]{%
  \begingroup
    \romannumeral-`\x % remove space at the beginning of \setcitestyle
    \setcitestyle{numbers}%
    \cite{#1}%
  \endgroup   
}

\newtheorem{theorem}{Theorem}[section]
\newtheorem{lemma}[theorem]{Lemma}

\begin{document}

\title{Frequency-encoded linear cluster states with coherent Raman photons}

\author{Dale Scerri}
\email{ds32@hw.ac.uk}
\author{Ralph N. E. Malein}
\author{Brian D. Gerardot}
\author{Erik M. Gauger}
\email{e.gauger@hw.ac.uk}
 
\affiliation{SUPA, Institute of Photonics and Quantum Sciences, Heriot-Watt University, EH14 4AS, United Kingdom.}
\date{\today}

\begin{abstract}
Entangled multi-qubit states are an essential resource for quantum information and computation. Solid-state emitters can mediate interactions between subsequently emitted photons via their spin, thus offering a route towards generating entangled multi-photon states. However, existing schemes typically rely on the excitation-relaxation of the emitter, resulting in single photons limited by the emitter's radiative lifetime, suffering from considerable practical limitations, for self-assembled quantum dots most notably the limited spin coherence time due to Overhauser magnetic field fluctuations. 
We here propose an alternative approach based on a spin-$\Lambda$ system that overcomes the limitations of previous proposals. Studying the example of  spin-flip Raman scattering of self-assembled quantum dots in Voigt geometry, we argue that weakly driven hole spins constitute a promising platform for the practical generation of frequency-entangled photonic cluster states.
%We here propose an alternative approach of employing spin-flip Raman scattering events of self-assembled quantum dots in Voigt geometry. We argue that weakly driven hole spins constitute a promising platform for the practical generation of frequency-entangled photonic cluster states.
\end{abstract}

\maketitle

\section{Introduction}

Robust highly-entangled `cluster' states are of paramount importance for measurement-based quantum computation \cite{Briegel2001b, Briegel2001a, Raussendorf2007, OBrien2009}. The experimental challenges of obtaining high-dimensional cluster states can be significantly reduced by probabilistically `fusing' qubits from adjacent 1D linear cluster (LC) states \cite{TerryFusion2005, Herrera2010, Weinstein2011}, or `glueing' together micro-clusters \cite{Nielsen2004}. Several platforms for generating photonic LC states have been proposed, varying from condensed matter emitters such as quantum dots \cite{Terry2009,  Denning2017, Barrett2005, Lin2008, Herrera2010, Schwartz2016} and crystal defects \cite{Economou2016, Barrett2005} to parametric downconversion \cite{Vallone2007, Zou2005}, all presenting their own sets of advantages and challenges. Solid-state-based protocols typically make use of pulsed excitations to drive optical transitions in a matter qubit to entangle the emitter's spin degree of freedom with the polarisation of subsequently emitted photons. Encouragingly, a photonic LC of length two (LC$_2$) has recently been demonstrated experimentally, showing that the entanglement in this setup could persist for up to five consecutively emitted photons\cite{Schwartz2016}.
\begin{figure}[t!]
\centering
\includegraphics[width=\linewidth]{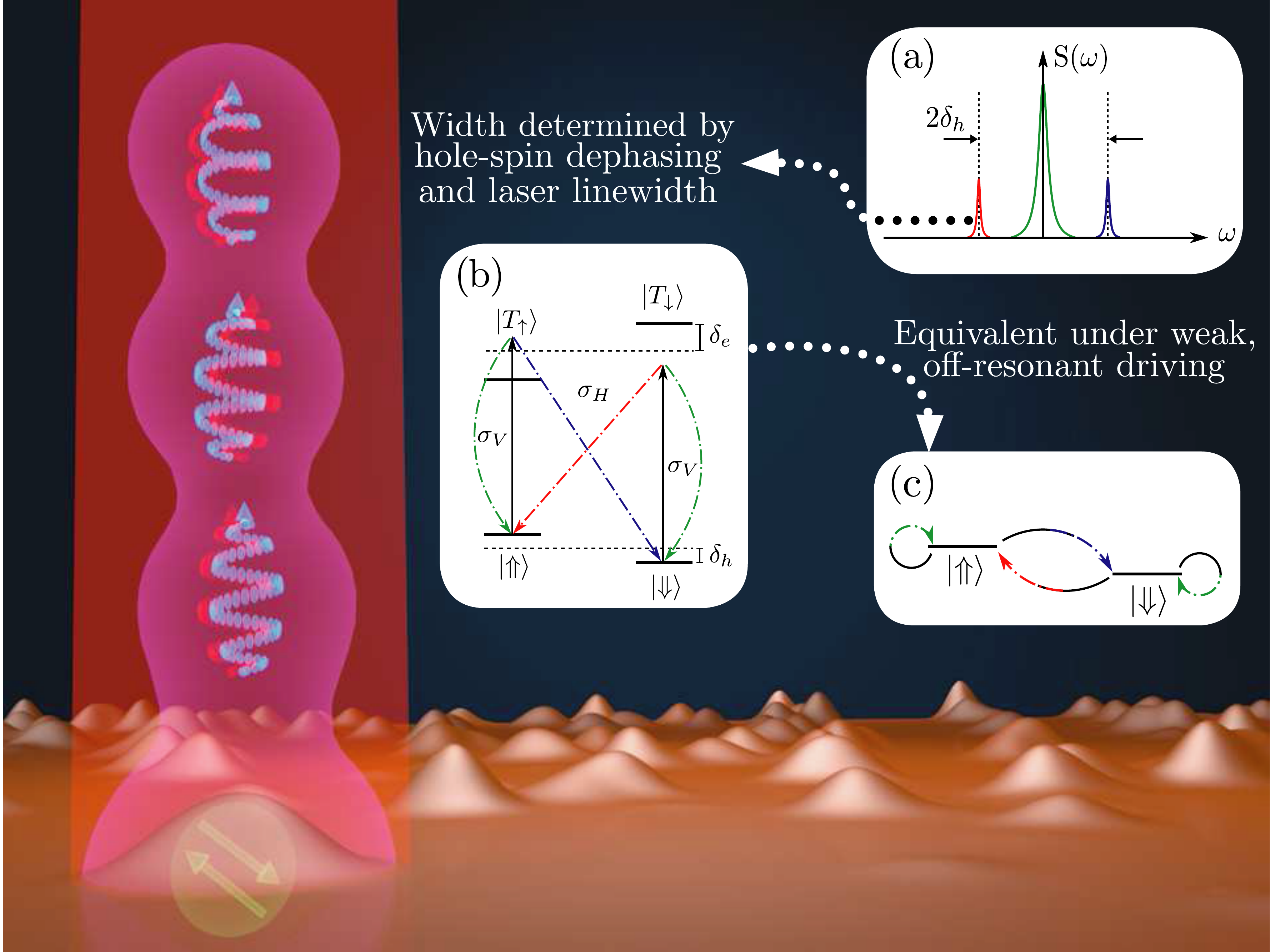} 		
\caption{\textbf{Background}: artistic depiction of our protocol. \textbf{Inset a)}: Schematic of the emission spectrum showing the presence of the Raman sidebands. \textbf{Inset b)}: Schematic illustration of the scattering processes involving the two ground hole-spin states. The black arrows denote the laser driving on resonance with the unperturbed transitions (dashed lines), whereas the green, red and blue arrows denote the Rayleigh, red detuned and blue detuned events, respectively. \textbf{Inset c)}: Simple schematic of the scattering processes involved in the weak, detuned driving limit.}
\label{Fig1}
\end{figure}

Whilst conceptually elegant and ostensibly deterministic, real-world imperfections pose significant barriers to the experimental realisation of protocols such as the ones introduced by Refs.~\onlinecite{Terry2009, Economou2010, Denning2017, Terry2018}. For the III-V platform, these include phonon-dephasing of excited states \cite{Iles-Smith2017}, modified selection rules as a consequence of hole mixing as well as a transverse (Voigt) component of the Overhauser field \cite{Loss2002, Testelin2009, Hansom2014, Malein2016}, and limited spin lifetimes due to Overhauser field fluctuations \cite{Loss2002, Merkulov2002, Braun2002, Chekhovich2013, Urbaszek2013}. Decoupling techniques \cite{Viola1999, Witzel2007, Zhang2007, Uhrig2007, Uhrig2008, Bluhm2010, Stockill2016} and control of the nuclear environment \cite{Eble2006, Petta2008, Urbaszek2013, Majcher2017} overcome the latter but provide no remedy for other error sources. Shortcomings of real quantum dots thus put a limit to the size of cluster state achievable and render genuinely deterministic operation impractical for the current experimental state-of-the-art.

In contrast to direct pulsed excitation, we here propose employing a weak (sub-saturation) continuous wave (c.w.) laser to drive the Zeeman-detuned transitions of a hole-spin for entangling the spin with the frequency of Raman scattered photon \footnote{The small Rabi energies entails that any dephasing due to the optical AC-Stark shift is negligible, although in principle AC-Stark shift tuning could be employed to significantly reduce the dephasing due to charge noise \cite{Ramsay2016}.}. We show such a setup overcomes the experimental barriers suffered by previous schemes, which rely on excitation and relaxation of the emitter: in particular, our protocol is impervious to phonon dephasing, robust against fluctuations of the Overhauser field, and unaffected by heavy-hole (hh) light-hole (lh) mixing. This comes at the cost making the protocol probabilistic, however, we show that LC states of sufficient length to serve as building blocks for fusion \cite{TerryFusion2005} can be produced at high rates and fidelity based on current experimental capabilities. Furthermore, extended versions of our protocol mitigating its probabilistic limitations (whilst keeping its robustness) are possible (see Sec.~\ref{app:Deter} of the Appendix). Our work thus shows that the significant divide between elegant theoretical proposals and experimental progress in the generation of linear cluster states  can be overcome. The approach we present has scope for extension to other quantum photonic platforms sharing a similar $\Lambda$-structure, including defects in wide-bandgap semiconductors\cite{Becker2016, Yale2013} and  superconducting artificial atoms\cite{Novikov2015, Liu2016, Premaratne2017}. Mathematical detail and extensions of the main protocol can be found in the Appendix Sec.~\ref{app:proof}.
\begin{figure}[t!]
\centering
\includegraphics[width=\linewidth]{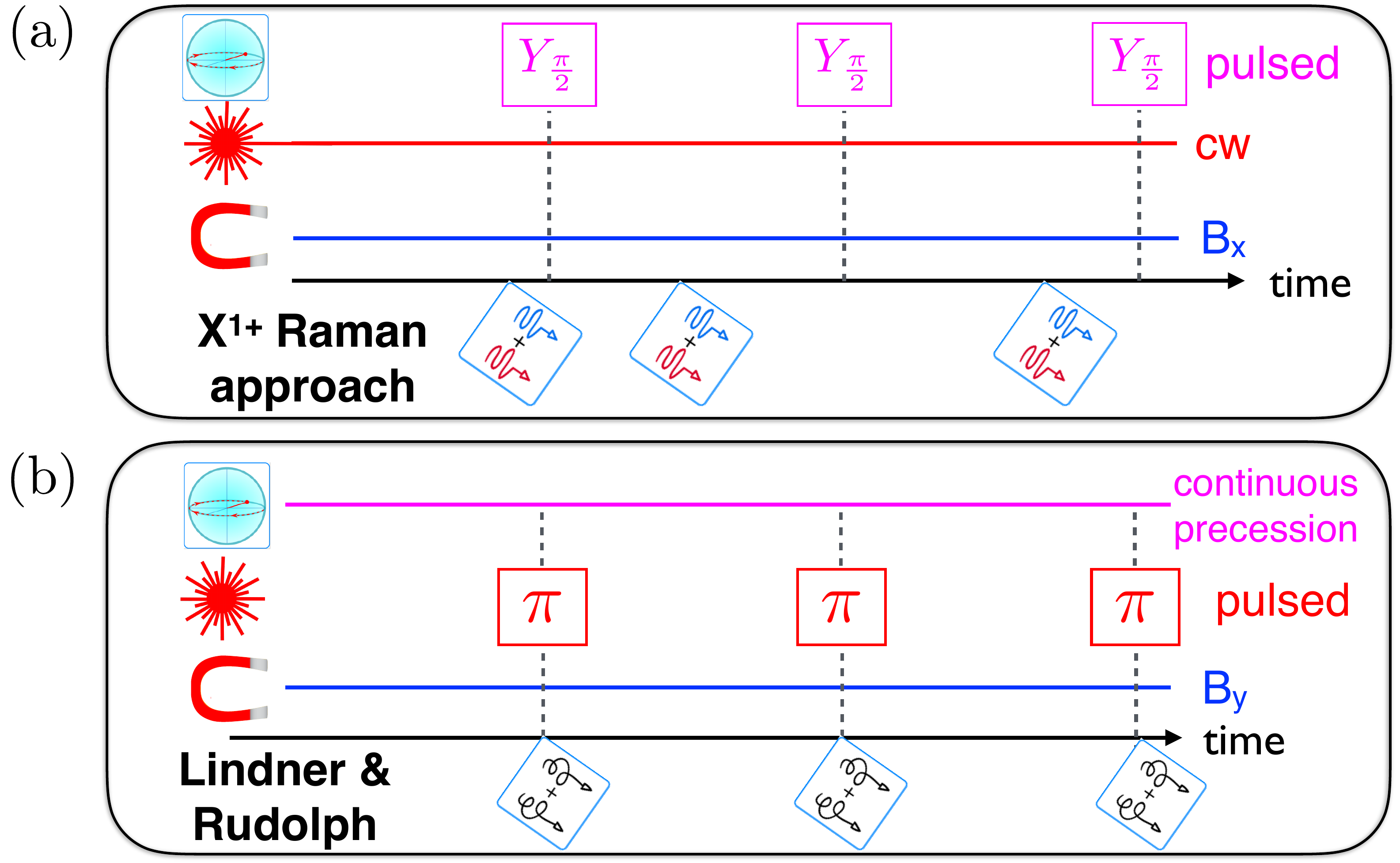} 
\caption{\textbf{a)} Schematic representation of our protocol. The spin precesses in a constant magnetic field in Voigt geometry. Driven weakly and off-resonantly, the hole-spin scatters Raman-detuned photons at random intervals. The timing between Y-pulses $T_B$ should be chosen so as to maximise the probability of getting a single scattering event between the pulses. \textbf{b)} Schematic of the original Lindner and Rudolph proposal for comparison. Instead of a gated Y-rotation, an external field in Voigt geometry causes the spin to precess continuously, with optical $\pi$-pulses applied at the appropriate times to excite the emitter.}
\label{timeline}
\end{figure}

\section{Model}

Despite its many attractive features for quantum metrology and quantum information\cite{Loss1998, Merkulov2002}, the spin of an electron trapped in an epitaxial quantum dot suffers from rapid ensemble dephasing due to the hyperfine interaction with $\sim 10^4 - 10^6$ randomly fluctuating nuclear spins of the host material. This typically results in a loss of coherence on the order of nanoseconds\cite{Merkulov2002, Braun2002, Chekhovich2013, Malein2016}. By contrast, the $p$-orbital-like wavefunction of hole spin states vanishes at the location of the nuclear spins, which suppresses the Fermi-contact interaction, leaving only the much weaker dipole-dipole interaction as the main source of dephasing \cite{Fischer2008, Testelin2009, Chekhovich2011, Fallahi2010}. Strain lifts the degeneracy of the $J = 3/2$ hole states, resulting in energetically split heavy ($J_z = \pm 3/2$) and light ($J_z = \pm 1/2$) holes; the former being closer to the valence band edge (see Fig.~\ref{Fig1}). Rashba or Dresselhaus spin-orbit coupling may may play a role in limiting factor for the application of these hole spins in quantum information, which were shown, both theoretically \cite{Bulaev2005} and experimentally \cite{DeGreve2011}, to limit the spin relaxation rate. However, we note that this spin-orbit coupling is still more detrimental to electron than hole spins \cite{Bulaev2005}.
\begin{figure}[t!]
\centering
\includegraphics[width=\linewidth]{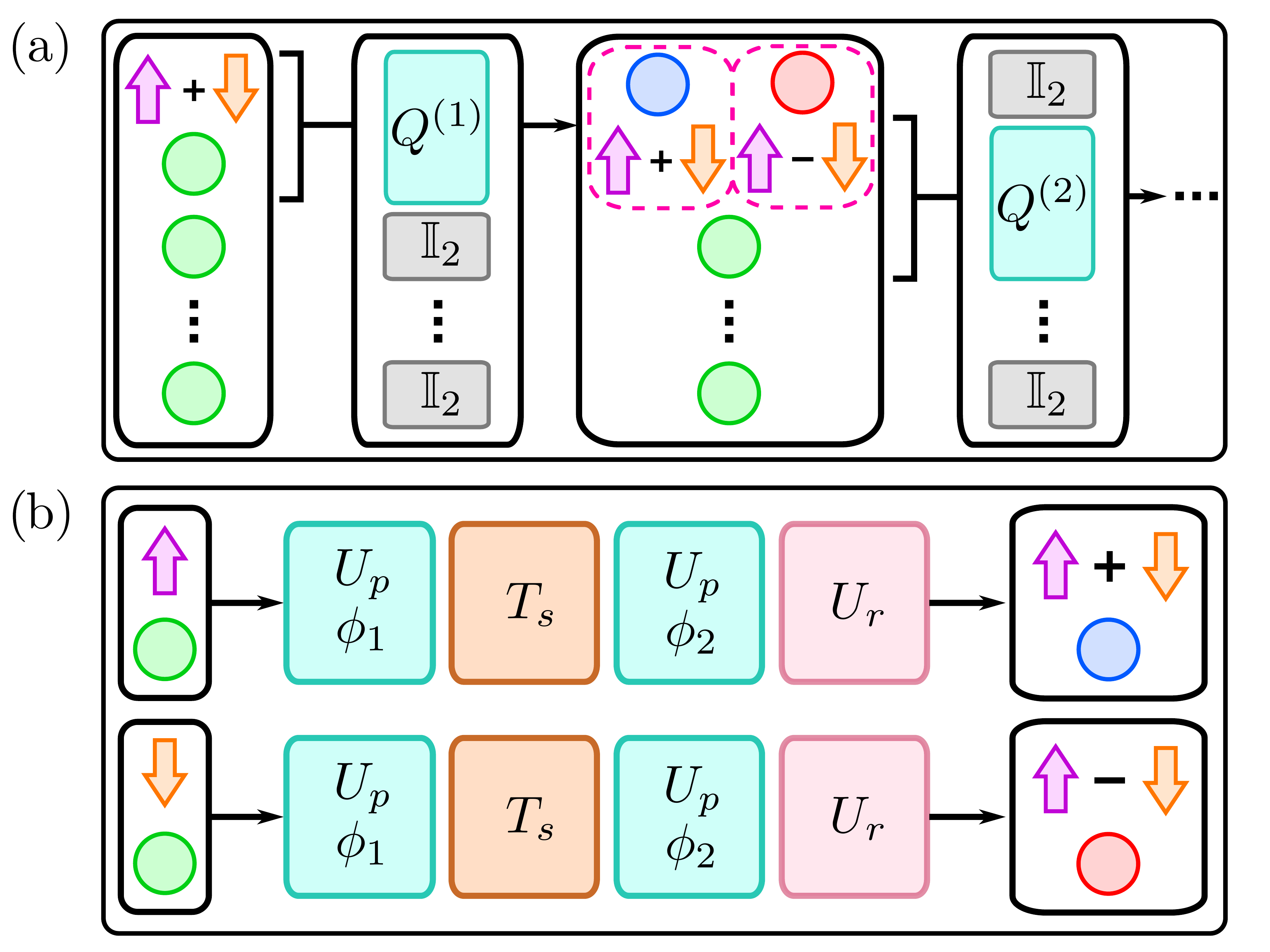} 
\caption{\textbf{a)}: Diagrammatic representation of the spin-photon entangling process for the first emitted photon. The initial spin state ($\Ket{\Ua} + \Ket{\Da}$) and first laser photon to be scattered (upmost green circle) undergo a joint transformation $Q^{(1)}$, resulting either in a red or a blue-detuned Raman photon that is entangled with the hole spin. $Q^{(2)}$ includes the second Raman process and entangles the newly scattered with the previous photon. \textbf{b)}: Breakdown of the $Q^{(i)}$ operation through its action on spin basis states: the sequence of operations transforms includes two periods of free spin precession $U_p$, the Raman scattering process $T_S$, and a $\pi / 2$ $Y$-rotation $U_r$.
A full matrix representation of $Q^{(i)}$ is given in Appendix Sec.~\ref{app:matop}.}
\label{protocol}
\end{figure}

However, chiefly due to strain anisotropy in the QD, a finite admixture of these states is always present (the effects on hole-based multi-photon entanglement schemes are briefly discussed in the Sec.~\ref{app:mix} of the Appendix). In the following, we denote the (Zeeman) spin state of the heavy hole as $\Ket{\Ua}$ and $\Ket{\Da}$ whereas the electron spin states are $\Ket{\ua}$ and $\Ket{\da}$. In this notation, the positively charged X$^{1+}$ transition $\Ket{\Ua} \leftrightarrow \Ket{T_\ua} = \Ket{\Ua \Da, \ua}$ couples to $\sigma^-$ polarised light and $\Ket{\Da} \leftrightarrow \Ket{T_\da} = \Ket{\Ua \Da, \da}$ to $\sigma^+$ light. In the presence of an external magnetic field in Voigt geometry, the otherwise dipole-forbidden diagonal Raman transitions are unlocked~(see Fig.~\ref{Fig1})\cite{Emary2007}. For weakly off-resonantly driven hole spins, the width of these Raman transitions is solely limited by the laser linewidth and ground state spin dephasing \cite{Imamoglu2009,Sun2016}, making them attractive candidates for single photon sources, as well as being attractive spin-spin qubit entanglers due to the spin's rich level scheme and selection rules \cite{Delteil2015, Stockill2017}.

Wishing to exploit such Raman photons for LC generation we consider a self-assembled quantum dot in the Voigt geometry, with the applied magnetic field $B$ strong enough to dominate over nuclear Overhauser field fluctuations (see Sec.~\ref{app:over} of the Appendix). The applied $B$-field (w.l.o.g. along the $x$-axis) then defines the basis of spin eigenstates. We also include a c.w.~laser field that is resonant with the unperturbed transition of the QD (Fig.~\ref{Fig1}a). In a frame rotating with the laser frequency (after performing the RWA), the QD Hamiltonian in the Zeeman basis reads
\begin{align}
\begin{split}
H =& \delta_h \left(\Ket{\Ua} \Bra{\Ua} - \Ket{\Da} \Bra{\Da} \right) +  \delta_e \left(\Ket{T_\da} \Bra{T_\da} - \Ket{T_\ua} \Bra{T_\ua} \right) \\
       -& \left( \frac{\Omega_H}{2}\Ket{T_\ua} \Bra{\Da} + \frac{\Omega_H}{2}\Ket{T_\da} \Bra{\Ua}  + \frac{\Omega_V}{2}\Ket{T_\da} \Bra{\Da}  \right. \\
       +& \left.  \frac{\Omega_V}{2}\Ket{T_\ua} \Bra{\Ua} + \mathrm{H.c.} \right) ~,
\end{split}
\end{align}
where $\delta_{e,h}$ are the electron and hole Zeeman splittings, respectively, $\Omega_{H / V}$ are the Rabi frequencies for the horizontally/vertically-polarised transitions, and $\mathrm{H.c.}$ denotes the Hermitian conjugate. We simulate the scattering events via Monte Carlo trajectories with jump operators for all allowed transitions, occurring with equal rates $\gamma$. This results in an  effective (non-Hermitian) Hamiltonian $H_{eff} = H - \frac{i \hbar}{2} \gamma \sum_n C_n^\dagger C_n$, where the sum goes over the collapse operators \cite{Johansson2012, Johansson2013}. This non-unitary evolution of the system generates photons outside of the QD's Hilbert space, which build the LC states we are interested in. More specifically, each `experiment' is simulated as a quantum jump simulation, where an LC$_n$ state is successfully measured if the correct $n$ scattering events occur within the designated time-bins. The success rate is then calculated by averaging over the results.

\section{Protocol}

Fig.~\ref{Fig1}b shows that the emission of blue and red-detuned Raman spin-flip photons from a single quantum dot must alternate, provided that the scattering rate is faster than the hole spin-flip time. We build on this correlation between spin and photon colour to develop a protocol for generating an entangled LC state (filtering out Rayleigh scattered photons via their orthogonal polarisation). As an intrinsic drawback of Raman spin-flips, the time at which a photon is scattered is not known prior to its detection. In the following, we assume that there is exactly one Raman scattering event per time-bin $T_B$ (albeit at a random time within the bin, see Fig.~\ref{timeline}.). The overall probability and ways of circumventing this  limitation\footnote{In practice, this assumption limits the size of the LCs that can be produced in this approach to less than ten.} will be discussed later. Fig.~\ref{protocol} contains a diagrammatic representation of a successful run of our protocol. Let us trace the evolution of the joint spin-photon-state step by step: we start with the hole spin initialised in the superposition state $\Ket{\Ua}+\Ket{\Da}$ (ignoring normalisation factors) and precessing at its Larmor frequency. Let the accumulated phase prior to the first scattering event be $\phi_1 = \delta_h \tau_1$ (denoted by the matrix $U_p(\phi_1)$ in Fig.~\ref{protocol}), then a Raman spin flip ($T_s$ in Fig.~\ref{protocol}) evolves the state to
\begin{equation}
\mathrm{e}^{-i \frac{\phi_1}{2}}\Ket{\Ua} + \mathrm{e}^{i \frac{\phi_1}{2}}\Ket{\Da} \rightarrow \mathrm{e}^{-i \frac{\phi_1}{2}}\Ket{\Da B_1} + \mathrm{e}^{i \frac{\phi_1}{2}}\Ket{\Ua R_1} ~,
\end{equation}
where the labels $B_1 (R_1)$ inside the ket denote the first emitted blue (red) photon. A subsequent period of free precession $\tau_2 = T_B - \tau_1$ until the end of the time-bin $T_B$ results in a phase $\phi_2 = \delta_h \tau_2$. We now apply a $\pi / 2$ $Y$-rotation ($U_r = Y_\frac{\pi}{2}$ in Fig.~\ref{protocol}), yielding the state
\begin{align}
\begin{split}
&\mathrm{e}^{-i \frac{\chi_1}{2}}\Ket{\Ua B_1} +\mathrm{e}^{-i \frac{\chi_1}{2}}\Ket{\Da B_1} \\
+ &\mathrm{e}^{i \frac{\chi_1}{2}}\Ket{\Ua R_1} - \mathrm{e}^{i \frac{\chi_1}{2}}\Ket{\Da R_1}~,
\end{split}
\end{align}
where $\chi_1 \coloneqq \phi_1 - \phi_2$. The next Raman scattering event will have been preceded by another spin precession angle $\phi_3$ resulting in
\begin{align}
\begin{split}
&\mathrm{e}^{-i \frac{\phi_3}{2}}\mathrm{e}^{-i \frac{\chi_1}{2}}\Ket{\Da B_1 B_2} +\mathrm{e}^{i \frac{\phi_3}{2}}\mathrm{e}^{-i \frac{\chi_1}{2}}\Ket{\Ua B_1 R_2}\\
+&\mathrm{e}^{-i \frac{\phi_3}{2}}\mathrm{e}^{i \frac{\chi_1}{2}}\Ket{\Da R_1 B_2} - \mathrm{e}^{i \frac{\phi_3}{2}}\mathrm{e}^{i \frac{\chi_1}{2}}\Ket{\Ua R_1 R_2}~.
\end{split}
\end{align}
The spin precesses further by $\phi_4$ before we apply the next $Y_\frac{\pi}{2}$ rotation, yielding
\begin{align}
\begin{split}
&\mathrm{e}^{-i \frac{\phi_3}{2}} \mathrm{e}^{i \frac{\phi_4}{2}} \mathrm{e}^{-i \frac{\chi_1}{2}}\Ket{\Da B_1 B_2} 
+\mathrm{e}^{i \frac{\phi_3}{2}}\mathrm{e}^{-i \frac{\phi_4}{2}}\mathrm{e}^{-i \frac{\chi_1}{2}}\Ket{\Ua B_1 R_2}\\
+&\mathrm{e}^{-i \frac{\phi_3}{2}}\mathrm{e}^{i \frac{\phi_4}{2}} \mathrm{e}^{i \frac{\chi_1}{2}}\Ket{\Da R_1 B_2} 
- \mathrm{e}^{i \frac{\phi_3}{2}}\mathrm{e}^{-i \frac{\phi_4}{2}}\mathrm{e}^{i \frac{\chi_1}{2}}\Ket{\Ua R_1 R_2} \\
\vspace{5mm}\\
\coloneqq &\mathrm{e}^{-i \frac{\chi_2}{2}}\mathrm{e}^{-i \frac{\chi_1}{2}}\Ket{\Da B_1 B_2} 
+\mathrm{e}^{i \frac{\chi_2}{2}}\mathrm{e}^{-i \frac{\chi_1}{2}}\Ket{\Ua B_1 R_2}\\
+&\mathrm{e}^{-i \frac{\chi_2}{2}}\mathrm{e}^{i \frac{\chi_1}{2}}\Ket{\Da R_1 B_2} 
- \mathrm{e}^{i \frac{\chi_2}{2}}\mathrm{e}^{i \frac{\chi_1}{2}}\Ket{\Ua R_1 R_2} ~.
\end{split}
\end{align}
Let us stop at this point and, for clarity, consider the resulting state without its free precession phases
\begin{equation}
\Ket{\Da B_1 B_2} + \Ket{\Ua B_1 R_2} + \Ket{\Da R_1 B_2} - \Ket{\Ua R_1 R_2}~.
\end{equation}
Using the photon qubit encoding $\Ket{B_i} = \Ket{1_i}$, $\Ket{R_i} = \Ket{0_i}$, the state following the final Y$_\frac{\pi}{2}$ rotation it is given by
\begin{align}
\begin{split}
&\Ket{\Ua 1_1 1_2} + \Ket{\Da 1_1 1_2} + \Ket{\Ua 1_1 0_2} - \Ket{\Da 1_1 0_2} \\
+ &\Ket{\Ua 0_1 1_2}  + \Ket{\Da 0_1 1_2} - \Ket{\Ua 0_1 0_2} + \Ket{\Da 0_1 0_2}~.
\end{split}
\end{align}
In Appendix Sec.~\ref{app:proof}, we show that, whether the spin is measured to be in the $\Ket{\Ua}$ or $\Ket{\Da}$ state, the resulting photonic state ($S^{(2)}_+$ or $S^{(2)}_-$, respectively) indeed corresponds to LC$_2$. Further, we show that the above protocol generalises trivially to the production of LC states of arbitrary length. Crucially, reintroducing the above precession phases keeps the state local-unitarily (LU) equivalent to LC$_2$. The phases become known post-measurement through the timestamps of the detection clicks, and in Sec.~\ref{app:Tom} of the Appendix, we discuss how to make allowances for them for a tomographic reconstruction of the LC state. \footnote{Time-stamping these photons does not impose any experimental challenges, as detector setups with $\approx 30$ ps readily resolve these phases due to the precession time being of the order of $\sim$2 ns for an external field of 100 mT.}
\begin{figure}[t!] 
\includegraphics[width=1\linewidth]{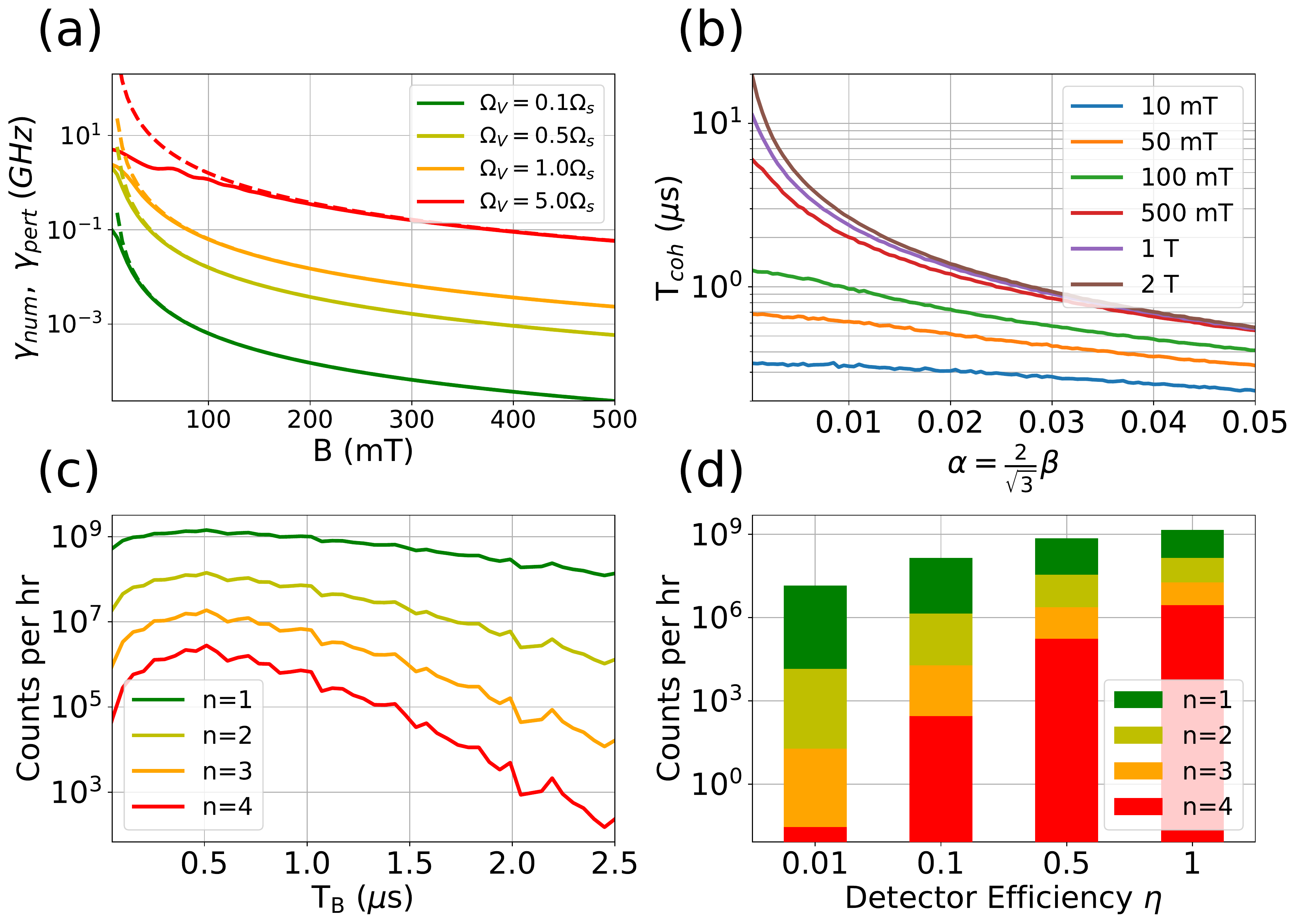}
\caption{\textbf{a)}: Perturbative calculation $\gamma_{pert}$ (dashed) and numerical value $\gamma_{num}$ (solid) of the Raman scattering rate as a function of $B$ for various driving strengths (from bottom to top: $\Omega_V = .1, .5, 1,$ and $5~\Omega_s$). \textbf{b)}: Coherence time for the pseudospin initially prepared perpendicular to the applied external magnetic field with mixing factor $\alpha = \frac{2}{\sqrt{3}} \beta$ for various external field strengths. The Overhauser field was taken to have a spread of 14mT (from bottom to top: $B = .01, .05, .1, .5, 1$ and 2 T). \textbf{c)}: Number of successful $n$-photon correlations per hour against $T_B$, with $\eta = 1$ for the ideal scenario $B = 100$mT, $\Omega = 0.2 {\gamma}/{\sqrt{2}}$, and $g^x_h = 0.1$ (from top to bottom: $n = 1, 2, 3$ and 4). \textbf{d)}: Success probabilities optimised for $T_B = 500$~ns [by minimising Eqn.~\eqref{FPeq}] against $\eta$, decreasing with increasing $n$.}
\label{probability}
\end{figure} 

\section{Results}

We now analyse the quality and success probability of our protocol. We begin with the rate for Raman scattering events followed by the success probability of a string of $n$ Raman photons with one per time-bin. Fig.~\ref{probability}a shows the Raman scattering rate and its dependence on both $B$ and $\Omega_V$. Comparison with numerical simulations shows that  this rate is well-approximated by the transition probability obtained by treating the weak driving field perturbatively to second order (see Appendix Sec.~\ref{app:pertrate})
\begin{equation}\label{gammapert}
\gamma_{pert} = \frac{1}{8} \frac{\Omega^2_V \gamma}{\Delta^2}~,
\end{equation}
provided $B \gtrsim 100$ mT and sub-saturation $\Omega_V \lesssim \gamma / \sqrt{2}$  (with $\gamma$ being the spontaneous emission rate), where $\Delta = \delta_e + \delta_h$. We proceed to determine the optimal duration $T_B$ (i.e.~the free precession time between $Y$-rotations) for maximising the probability of obtaining a single Raman event per time bin. Adopting $B = 100$~mT and $\Omega_V = 0.2 {\gamma}/{\sqrt{2}}$ (taking $\gamma = 1~$ns$^{-1}$), we calculate the number of successful trials with one Raman photon per $T_B$ interval (time interval between $U_r$ rotations in Fig.~\ref{protocol}) in $n$ successive time-bins. Fig.~\ref{probability}c illustrates the results of Monte-Carlo simulations using the QuTiP package \cite{Johansson2012, Johansson2013} for $n=1$ to $4$ scattering events, suggesting that $T_B \approx 0.5~\mathrm{\mu s}$ is close to optimal. We have the relation $P_s(n) = P_s(1)^n$ between the success probability for a single bin and that of $n$ bins.

Apart from addressing the possibility of having no Raman events within a time-bin, we also need to account for the possibility of `false-positives', i.e.~detecting only one of multiple Raman events occurring in a single time-bin, due to a photon detection efficiency $\eta < 1$\footnote{We assume $\eta$ is the probability of obtaining a detector click if a photon was produced by the QD, i.e.~it also includes any photon losses in the setup.)}. The probability of such $n$ photon false positives, $P_{fp}(n)$, is given by the simple relation:
\begin{align}\label{FPeq}
\begin{split}
P_{fp}(n) &= P_{nd}(n) \times P_d(1) \times P_s(n+1) \\
	       &=  {C}^{n+1}_n(1-\eta)^{n} \times \eta \times P_s(n+1)~,
\end{split}
\end{align}
where ${C}^{n+1}_n$ is the binomial coefficient, $P_d(n)$ [$P_{nd}(n)$] denotes the probability of detecting [not detecting] $n$ photons. We find that $T_B \approx 0.5~\mathrm{\mu s}$ remains optimal after taking this into account. Fig.~\ref{probability}d shows the rate of LC generation for $n=1$ to $4$ for different detector efficiencies. 

To demonstrate the robustness of our protocol against nuclear environment fluctuations, we calculate the fidelity between the state obtained with and without Overhauser field (both for the the same set of precession phases determined by randomly chosen scattering times). For a pure hh, only the $B^z_N$ Overhauser component perpendicular to the applied $B$-field affects the protocol [by randomly modifying direction and magnitude of the total $B$-field by $\arctan(B^{z}_N / B_{ext})$]. By contrast, a mixed hh--lh system suffers predominantly from the parallel $B^x_N$ component, to an extent determined by the mixing factor $\alpha$. This is also exemplified in a decreased spin coherence time from the ideal hh limit, as shown in Fig.~\ref{probability}b. Only considering this term, the following analytical expression (see Appendix Sec.~\ref{app:fidel}) captures the fidelity decay as a function of $T_B$:
\begin{equation}\label{fidelav}
\bar{\mathcal{F}}^{(1)} = \frac{1}{2} + \frac{\sqrt{2 \pi}}{4 T_B \delta B^x_N} \mathrm{erf}\left(\frac{T_B \delta B^x_N}{\sqrt{2}} \right)~,
\end{equation}
where $\bar{\mathcal{F}}^{(n)}$ denotes the average fidelity for a state of $n$ scattered photons (written for $n=1$ in Eq.~\eqref{fidelav} above), and $\delta B^x_N$ is the fluctuation in $B^x_N$. For a single scattered photon, we obtain $\mathcal{F}^{(1)}_{av} \to 1/2$ for large $T_B$ as expected. Not capturing decoherence due to $B^z_N$ fluctuations,  Eqn.~\eqref{fidelav} represents an upper bound on the maximally achievable fidelity in the case of finite hh-lh mixing. To fully account for the effects of the stochastically varying net $B$-field vector, we show numerically obtained\footnote{The numerical calculations were performed using the Overhauser ensemble-averaged matrix operations defined in Sec.~\ref{app:fidel} of the Appendix.} fidelity overlaps of desired vs the ensemble-average of realised LC$_4$ states in Fig.~\ref{Fplots}.  In the presence of the Overhauser field with fluctuations $\sim 14$mT, near unit fidelity remains possible in the region with (moderately) strong $B \gtrsim 0.4$~T and relatively short $T_B \lesssim 0.25$~$\mu$s (Fig.~\ref{Fplots}a). Conversely, large LC generation rates demand $0.5~\mathrm{\mu s} \lesssim T_B \lesssim 1$~$\mu$s and $B \lesssim 0.1$~T (Fig.~\ref{Fplots}b), so that a trade-off situation arises. Encouragingly, there is a wide middle-ground where high fidelity operation is possible at respectable rates.
\begin{figure}[t!] 
\includegraphics[width=1\linewidth]{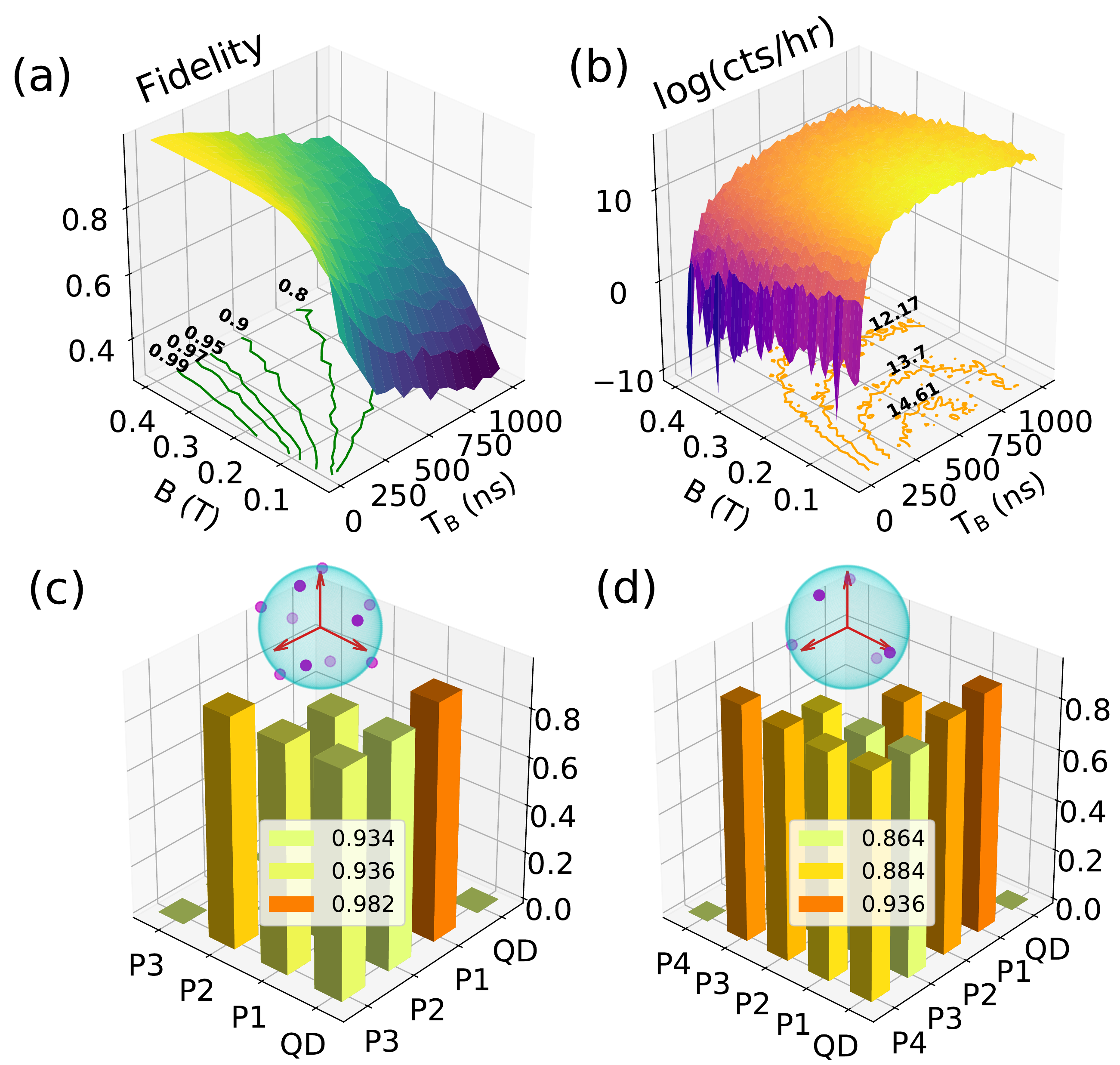}
\caption{\textbf{a)} Fidelity of the LC$_4$ state in the presence of the Overhauser field against applied field magnitude and single time-bin duration for a mixed hh--lh spin state. Overhauser fluctuations were 14~mT \cite{Malein2016}, with $g^x_h = 0.1$, $\alpha  = .01$ and a completely unpolarised spin bath. \textbf{b)} Natural logarithm of the success counts for a string of four photons. The overall detector efficiency was taken to be $\eta = 1$. The count rate increases with $T_B$ until probability of multiple events in a single bin becomes significant. An increasing B-field decreases the count rate as predicted from Eqn.~\eqref{gammapert}. {\textbf{c) and d)} Normalised LE for between pairwise combinations of a spin and 3 (panel c) or 4 (panel d) scattered photons, respectively. Due to computational constraints, we limited ourselves to ten (panel c) and five (panel d) uniformly distributed basis states on the Bloch sphere (with projectors shown in relevant insets).}} 
\label{Fplots}
\end{figure} 

Another important figure of merit of our protocol is the localisable entanglement (LE) \cite{Popp2005, Schwartz2016} between any two qubits of the LC state (including the spin). The LE represents the maximum negativity of the reduced density matrix of two qubits of interest (indexed $j$ and $k$), after all others have measured out projectively. Choosing the set of projectors $\mathcal{M} = \{ P_i : 1 \leq i \leq n,~i \notin \{j,k \} \}$ as our measurement defines an ensemble $\mathcal{E}_\mathcal{M} \coloneqq \{ p_{\mathcal{M},s}, \rho^{j,k}_{\mathcal{M},s} \}$, where $p_{\mathcal{M},s}$ is the probability of obtaining the two-spin density matrix $\rho^{j,k}_{\mathcal{M},s}$ for the outcome $\{s\}$ having measured the remaining $N-2$ qubits. The LE is then defined as the maximum negativity after averaging over all the outcomes for each measurement, that is
\begin{equation}
LE^\mathcal{N}_{j, k} = \max\limits_\mathcal{M} \sum_{s} p_{\mathcal{M},s}~\mathcal{N}(\rho^{j,k}_{\mathcal{M},s})~,
\label{eq:le}
\end{equation}
where $\mathcal{N}(\rho^{j,k}_{\mathcal{M},s})$ is the negativity of $\rho^{j,k}_{\mathcal{M},s}$. We choose a quasi-uniformly distributed basis on the Bloch sphere of each qubit  (see points in insets of Fig.~\ref{Fplots}c,d). The computational unwieldiness of Eqn.~(\ref{eq:le}) restricts the number of projectors, and we can only obtain a lower-bound of the true LE for LC$_{3,4}$ (Fig.~\ref{Fplots}c,d). Within the variance of the sample over which the optimisation was performed, the LE falls off with qubit distance, but encouragingly it remains remarkably high overall, and is thus unlikely to be a limiting factor in the length of the LC that could be generated using this protocol. 

\section{Overhauser field limitations}
\label{app:short}

The relatively short $T^*_2$ time of the electron spin due to the fluctuating nuclear environment constitutes a severe shortcoming of real quantum dot spins, putting a limit on the order of a few nanoseconds on any experiment relying on the coherence of this system. For the LR protocol \cite{Terry2009} one requires an external field of the order of $\sim50$mT along the Y direction in order to obtain a sufficient number of $Y$-gates for a multi photon LC$_{4 \geq n \geq 2}$ state within a few nanoseconds (assuming instantaneous excitation and radiative decay). Such an applied field, however, activates the previously dipole-forbidden transitions, degrading the correlations between the spin and emitted photons required for the LC state. Applying a strong field results in significant electron-spin precession between the pulsed excitation and spontaneous emission events, reducing the fidelity of the produced LC. By contrast, applying a weaker field limits the scalability of the protocol beyond a string of a couple of photons, as well as failing to screen the effects of the fluctuating Overhauser field. In short, the presence of the Overhauser field implies that the LR protocol would in practice need to be upgraded to incorporate dynamical decoupling and gated $Y$-rotations instead of relying on free spin precession.

One way to overcome some of these hurdles would be to adapt the LR protocol to a hole-spin system, having a longer dephasing time. However, due to the hole spins coupling weakly to external magnetic fields, the precession time would be much longer, requiring stronger fields to implement the $Y$-rotations, hence resulting in the same issue discussed above; namely, the undesirable dipole-forbidden transitions becoming accessible. Shorter coherence times in the presence of a weak external field and phonon sideband emissions (see below) would also be an issue in the hole-spin variant of the LR scheme. Hence, our scheme goes beyond a direct adaptation of the original LR scheme to the hole-spin platform, which would still suffer from most of the shortcomings of the original proposal.

Extending the promising dark exciton (DE) scheme\cite{Schwartz2016} beyond a couple of photons present similar experimental challenges: the finite radiative lifetime of the biexciton (BiE) $\tau_\mathrm{BiE} \approx 0.33$ns entails that the spin precesses by a non-negligible random amount both in the DE and BiE states, and this limits the purity of the photon polarisation state. Furthermore, the DE spin also suffers from environmental decoherence during its precession \cite{Schwartz2016}. It should be noted, however, that the dark exciton scheme proposed in Ref.~\citen{Schwartz2016} could be optimised (for example, by using Purcell enhancement) to improve scalability.

The elegant recently proposed scheme in Ref.~\citen{Denning2017} was designed to be robust against Overhauser fluctuations, provided the scattering events occur on a short enough timescale over which the Overhauser field can be assumed constant (so that only a global phase is gained in each trajectory). However, in this case an additional single photon source and high cooperativity is required, and any lifting of the selection rules (e.g. due to hole mixing, see below) will still impose practical limitations.

\section{Conclusion}

We have presented a novel scheme for generating frequency-encoded LC states, which could serve as a  stepping stone towards measurement-based quantum computation. Unlike current rival schemes, our protocol does not rely on the excitation and relaxation of the emitter, and is therefore only sensitive to ground-state hole-spin dephasing, at the cost of being limited by its intrinsic probabilistic nature. Based on experimentally informed properties of real epitaxial quantum dots, we have shown that LC states of sufficient length and high fidelity for fusion into larger cluster states can nevertheless be produced at respectable rates. In turn, this facilitates type-II fusing into 2D cluster states \cite{TerryFusion2005, Morley2017}. Our protocol takes full account of unmitigated Overhauser field fluctuations. It is inherently impervious to hole-mixing induced modifications of the optical selection rules, but, like other approaches, it stands to gain from dynamic decoupling. 

Whilst the probabilistic nature of the Raman scattering events limits our protocol as described in the main text to LC states of length $n < 10$, our approach can, in principle, be made deterministic. The most elegant way of achieving this would be to detect the presence of the Raman scattered photons without absorbing them or learning their frequency, however, this ability does not currently exist for optical photons, which is why we turn to observing the QD spin instead. Continuously monitoring whether a Raman spin-flip has happened, but without learning the spin state itself, requires the introduction of a secondary `ancilla'  quantum dot as a witness of the spin-photon entangling event. These extensions, discussed in more detail in Sec.~\ref{app:Deter} of the Appendix, make the Raman hole-spin emitter a viable, practical alternative in the quest for realising non-classical multi-photon states, and importantly one which can be straightforwardly implemented with current expertise and devices.

\begin{acknowledgments}

We thank David Gershoni, Emil Denning, and Jake Iles-Smith for insightful and stimulating discussions. D.S. thanks SUPA for financial support. We acknowledge support from EPSRC (EP/M013472/01) and the ERC (no. 307392).
B. D. G. thanks the Royal Society for a Wolfson Merit Award, and E. M. G. acknowledges support from the Royal Society of Edinburgh and the Scottish Government.

\end{acknowledgments}

\section*{Appendix}
\appendix

\section{Second-order perturbation rate}
\label{app:pertrate}

It can be easily shown that, after moving to a rotating frame with respect to the unperturbed transition frequency, the amplitude of the Raman-flip transition  $\Ket{\Da} \rightarrow \Ket{\Ua}$ is given by

\begin{align}\label{DUamp}
\begin{split}
\mathcal{T}_{\Da\rightarrow\Ua} = &\frac{\Bra{\Ua; \omega_{R}} H_I \Ket{T_\da; 0} \Bra{T_\da; 0} H_I \Ket{\Da; \omega_{Ray}}}{\hbar\Delta^{(1)}_1} \\
				  + & \frac{\Bra{\Ua; \omega_{Ray}} H_I \Ket{T_\ua; 0} \Bra{T_\ua; 0} H_I \Ket{\Da; \omega_{B}}}{\hbar\Delta^{(1)}_2} ~,
\end{split}
\end{align}
where $\Delta^{(1)}_1 = \delta_h + \delta_e$, $\Delta^{(1)}_2 = \delta_h - \delta_e$, $H_I$ is the light-matter interaction Hamiltonian (in this case between the spin and c.w. laser field), and  $\omega_R$, $\omega_B$ and $\omega_{Ray}$ are the red-, blue-detuned and Rayleigh scattered photon frequencies, respectively. The first term in Eqn.~\eqref{DUamp} gives the amplitude of a red Raman photon event: the system, initially in the $\Ket{\Da}$ state, scatters a $\sigma^V$ photon, after which the final state is given by $\Ket{\Ua; H}$ (that is, the system in the $\Ket{\Ua}$ state and a red-detuned Raman photon ($\sigma^H$ polarised) is scattered). Similarly, the $\Ket{\Ua} \rightarrow \Ket{\Da}$ transition giving rise to the blue-detuned photon scattering event occurs with amplitude

\begin{align}\label{UDamp}
\begin{split}
\mathcal{T}_{\Ua\rightarrow\Da} = &\frac{\Bra{\Da; \omega_{B}} H_I \Ket{T_\ua; 0} \Bra{T_\ua; 0} H_I \Ket{\Ua; \omega_{Ray}}}{\hbar\Delta^{(2)}_1} \\
				  + & \frac{\Bra{\Da; \omega_{Ray}} H_I \Ket{T_\da; 0} \Bra{T_\da; 0} H_I \Ket{\Ua; \omega_{R}}}{\hbar\Delta^{(2)}_2} ~,
\end{split}
\end{align}
where $\Delta^{(2)}_1 = -\delta_h - \delta_e$, $\Delta^{(2)}_2 = -\delta_h + \delta_e$.

The second term in each of the transition amplitudes does not contribute to the Raman processes, and vanish as the driving field can only drive vertically-polarised transitions. After performing the necessary solid angle integrals, we arrive at the scattering rate given by Eqn.~\eqref{gammapert} in the main text.

\section{Overhauser field for hole-spin systems}
\label{app:over}

Vanishing wavefunctions at the nuclear sites means that the Fermi-contact hyperfine term for the nuclear--hole spin interaction is effectively zero, leaving only the dipole-dipole interaction term as the dominant source of dephasing. For an idealised pure hh, this term is of Ising-nature, with just the ZZ component being present. In most epitaxially grown QDs, however, some degree of hh $\Ket{J; J_z} = \Ket{3/2; \pm 3/2}$ and lh $\Ket{J; J_z} = \Ket{3/2; \pm 1/2}$ mixing is always present \cite{Prechtel2016, Testelin2009}, breaking the Ising-like nature of the dipole-dipole term and introducing XX and YY terms in the Hamiltonian. This means that the eigenstates of the Hamiltonian are no longer given separately by the hh or lh states, but a linear combination of both (the consequences of this mixing in quantum dot-based LC protocols is further discussed in Sec.~\ref{app:mix} of the Appendix). Without going into too much detail, the hyperfine coupling Hamiltonian for the hh--lh states is given by:

\begin{equation}
H^{dd}_{hf} = V \sum_j C_j |\Psi(\mathbf{R}_j)|^2 \left[ \alpha (I^j_x S_x + I^j_y S_y) + I^j_z S_z \right]
\end{equation}
where $C_j$ are dipole-dipole hyperfine constants, $V$ is the unit cell volume, and $\alpha =  \frac{2}{\sqrt{3}}|\beta|$ is a parameter depending on the deformation potentials for the valence band, and the strain tensor \cite{Prechtel2016,Testelin2009}. In the `frozen-fluctuation' model \cite{Merkulov2002}, this results in an effective magnetic field with mean $\Braket{\mathbf{B}_N} = (\Braket{B^{x}_N}, \Braket{B^{y}_N}, \Braket{B^{z}_N})$ (which, due to the finite size of the spin bath, is not necessarily zero), and a fluctuation $\delta \mathbf{B}_N = (\delta B^{x}_N, \delta B^{y}_N, \delta B^{z}_N)$ (which is the source of the spin's loss of coherence), and is assumed to follow normal statistics\cite{Testelin2009}:

\begin{align}
\begin{split}
P(\mathbf{B}_N)  &= \left(\frac{1}{2 \pi}\right)^{\frac{3}{2}} \frac{1}{\delta B^{\| 2}_N \delta B^{\bot}_N} \\
\vspace{5mm}\\
& \times \exp\left[ - \frac{\Delta B^{x\;2}_N}{2 \: \delta B^{\| 2}_N} - \frac{\Delta B^{y\;2}_N}{2 \: \delta B^{\| 2}_N} - \frac{\Delta B^{z\;2}_N}{2 \: \delta B^{\bot 2}_N}  \right] ~,
\end{split}
\end{align}
where $\Delta B^{i}_N = B^{i}_N - \Braket{B^{i}_N}$, $\delta B^\bot_N = \delta B^z_N$ and $\delta B^\|_N \coloneqq \delta B^x_N = \delta B^y_N  = \alpha \: \delta B^\bot_N $. Experimentally, Overhauser field fluctuations of 10--30mT have been measured \cite{Chekhovich2013, Urbaszek2013}, putting a lower-bound on the applied external field required to screen these fluctuations.
%
%\begin{figure}[t!]
%\centering
%\includegraphics[width=\linewidth]{timeline.pdf} 
%\caption{Schematic representation of our protocol. The spin precesses in a constant magnetic field in Voigt geometry. Driven weakly and off-resonantly, the hole-spin scatters Raman-detuned photons at random intervals. The timing between Y-pulses $T_B$ should be chosen so as to maximise the probability of getting a single scattering event between the pulses.}
%\label{timeline}
%\end{figure}

\section{Matrix operations}
\label{app:matop}

Consider a single scattering process that can be described by the action of the product of matrices:

\begin{align}\label{scatmats}
\begin{split}
\Ket{\Ua}\Ket{Ray_k} \rightarrow \hspace{1mm} &\mathrm{e}^{-i \frac{\phi^{(k)}_1}{2}}  \mathrm{e}^{i \frac{\phi^{(k)}_2}{2}} \left(\Ket{\Ua} + \Ket{\Da} \right) \Ket{B_k} \\
=& U_r U_p(\phi^{(k)}_2) T^{(k)}_s U_p(\phi^{(k)}_1) \Ket{\Ua}\Ket{Ray_k} \\
=& Q^{(k)} \Ket{\Ua}\Ket{Ray_k} ~,\\
\vspace{5mm}\\
\Ket{\Da}\Ket{Ray_k} \rightarrow \hspace{1mm} &\mathrm{e}^{i \frac{\phi^{(k)}_1}{2}}  \mathrm{e}^{-i \frac{\phi^{(k)}_2}{2}} \left(\Ket{\Ua} - \Ket{\Da} \right) \Ket{R_k} \\
=& U_r U_p(\phi^{(k)}_2) T^{(k)}_s U_p(\phi^{(k)}_1) \Ket{\Da}\Ket{Ray_k} \\
=& Q^{(k)} \Ket{\Da}\Ket{Ray_k} ~,\\
\end{split}
\end{align}
where $U_p(\phi^{(k)}_{1,2})$ is the free spin precession transformation before ($\phi^{(k)}_1$) and after ($\phi^{(k)}_2$) the $k^\mathrm{th}$ scattering event (prior to the Y$_\frac{\pi}{2}$ rotation), with the resulting matrix of events being $Q^{(k)} \coloneqq U_r U_p(\phi^{(k)}_2) T^{(k)}_s U_p(\phi^{(k)}_1)$. The scattering matrix $T^{(k)}_s$ is given by

\begin{align}
\begin{split}
T^{(k)}_s &=  \left(\begin{array}{c|c}0 & T^{(k)}_{R} \\\hline T^{(k)}_{B} & 0\end{array}\right)~,
\end{split}
\end{align}
with $T^{(k)}_{R}$ and $T^{(k)}_{B}$ written in the basis $\{\Ket{B_k},\Ket{R_k},\Ket{Ray_k}\}$, which simultaneously flips the spin state $\Ket{\Ua} \leftrightarrow \Ket{\Da}$, and applies the local transformations 

\begin{align}
\begin{split}
T^{(k)}_{B} &: \Ket{Ray_k} \mapsto \Ket{B_k} \\
T^{(k)}_{R} &: \Ket{Ray_k} \mapsto \Ket{R_k}~,
\end{split}
\end{align}
where we have omitted the unaffected photon states for brevity. Hence $T^{(k)}_{B}$ and $T^{(k)}_{R}$ take the form:

\begin{align}
\begin{split}
T^{(k)}_{R} &= \mathbb{I}_3^{\bigotimes_{k-1}} \otimes \left(\begin{array}{ccc}0 & 0 & 0 \\0 & 0 & 1 \\0 & 0 & 0\end{array}\right) \otimes \mathbb{I}_3^{\bigotimes_{n-k}} \\
T^{(k)}_{B} &= \mathbb{I}_3^{\bigotimes_{k-1}} \otimes \left(\begin{array}{ccc}0 & 0 & 1 \\0 & 0 & 0 \\0 & 0 & 0\end{array}\right) \otimes \mathbb{I}_3^{\bigotimes_{n-k}}~,
\end{split}
\end{align}
and $U_r$ and $U_p(\phi)$ are simply given by given by

\begin{align}\label{matrices}
\begin{split}
U_r &= \exp\left(i \frac{\pi}{4} \sigma_y\right) \otimes \mathbb{I}_3^{\bigotimes_{n}} \\
U_p(\phi) &= \left(\begin{array}{cc}\mathrm{e}^{-i \frac{\phi}{2}} & 0 \\0 & \mathrm{e}^{i \frac{\phi}{2}}\end{array}\right) \otimes \mathbb{I}_3^{\bigotimes_{n}}~,
\end{split}
\end{align}
where the first matrices act on the spin state and have been written in the $\{ \Ket{\Ua},\Ket{\Da} \}$ basis. Unfortunately, the matrix product describing $n$-photon scattering events becomes unwieldy with increasing $n$. In Appendix Sec.~\ref{app:proof}, however, we show that this protocol does indeed generalise to a LC$_n$ state, up to free precession phases.

\section{Generalisation to $n$-photons}
\label{app:proof}

\subsection{Preliminary lemmas}

In this section, we will show that the general form of the $n$-photon state $S^{(n)}$ obtained using our protocol can be written recursively (where we have suppressed the ket representation for these states for ease of notation). In fact,

\begin{lemma}\label{lemma1}
$\forall n \in \mathbb{N}$, the $n$-photon state $S^{(n)}$ can be decomposed into the recursive relations

\begin{align}
\begin{split}
S^{(n)}_+ = S^{(n-1)}_+ \Ket{1_n} + S^{(n-1)}_- \Ket{0_n} ~, \\
S^{(n)}_- = S^{(n-1)}_+ \Ket{1_n} - S^{(n-1)}_- \Ket{0_n}~,
\end{split}
\end{align}
depending whether the spin is measured to be in the $\Ket{\Ua}$ or $\Ket{\Da}$ state, respectively.  
\end{lemma}

\begin{proof}
We will, w.l.o.g., ignore the spin precession, although the proof is the same for the general case:

\vspace{5mm}
\textit{Basis case}: For $j=1$, $S^{(1)}_+ = \Ket{1_1} + \Ket{0_1}$ and $S^{(1)}_- = \Ket{1_1} - \Ket{0_1}$. After the next scattering event, we get

\begin{align}
\begin{split}
S^{(2)}_+ &= \Ket{1_1 1_2} + \Ket{1_1 0_2} + \Ket{0_1 1_2} - \Ket{0_1 0_2} \\
&= (\Ket{1_1} + \Ket{0_1})\Ket{1_2} + (\Ket{1_1} - \Ket{0_1})\Ket{0_2} \\
&= S^{(1)}_+ \Ket{1_2} + S^{(1)}_- \Ket{0_2} ~.
\end{split}
\end{align}
Similarly,

\begin{align}
\begin{split}
S^{(2)}_- &= \Ket{1_1 1_2} - \Ket{1_1 0_2} + \Ket{0_1 1_2} + \Ket{0_1 0_2} \\
&= (\Ket{1_1} + \Ket{0_1})\Ket{1_2} - (\Ket{1_1} - \Ket{0_1})\Ket{0_2} \\
&= S^{(1)}_+ \Ket{1_2} - S^{(1)}_- \Ket{0_2} ~.
\end{split}
\end{align}

\vspace{5mm}
\textit{Induction step}: Assume statement holds for $j=n$, and consider the $(n+1)^\mathrm{th}$ scattering event:

\begin{align}
\begin{split}
S^{(n+1)}_+ = &~U_r T^{(n+1)}_{scat} ( \Ket{\Ua}  S^{(n)}_+ + \Ket{\Da} S^{(n)}_-  ) \Ket{Ray_{n+1}} \\
                    = &~(\Ket{\Ua} + \Ket{\Da}) S^{(n)}_+ \Ket{1_{n+1}} + (\Ket{\Ua} - \Ket{\Da}) S^{(n)}_- \Ket{0_{n+1}} \\
                    = &~\Ket{\Ua} (S^{(n)}_+ \Ket{1_{n+1}} + S^{(n)}_- \Ket{0_{n+1}}) \\
                       &~+ \Ket{\Da} (S^{(n)}_+ \Ket{1_{n+1}} - S^{(n)}_- \Ket{0_{n+1}})~.
\end{split}
\end{align}
Therefore  $S^{(n+1)}_+ = S^{(n)}_+ \Ket{1_{n+1}} + S^{(n)}_- \Ket{0_{n+1}}$ and $S^{(n+1)}_- = S^{(n)}_+ \Ket{1_{n+1}} - S^{(n)}_- \Ket{0_{n+1}}$, so the statement holds $\forall n \in \mathbb{N}$.
\end{proof}

It is then easy to see that we also have that

\begin{lemma}\label{lemma2}
\begin{equation}
\sigma^{(n)}_z S^{(n)}_\pm = -S^{(n)}_\mp \quad \forall n \in \mathbb{N}~,
\end{equation}
\end{lemma}

which we shall use to prove that the $n$-photon state we generate is indeed a linear cluster state.

\subsection{Equivalence to LC$_n$ states}

In order to show that the $S^{(n)}_\pm$ states are indeed LC$_n$s, we have to show that they both satisfy the set of eigenvalue equations

\begin{equation}\label{conditions}
K^{(a)}_n S^{(n)}_\pm = (-1)^{k^{(a)}_\pm} S^{(n)}_\pm~,
\end{equation}
with

\begin{equation}
K^{(a)}_n = \sigma^{(a)}_x \bigotimes_{b \in N(a)} \sigma^{(b)}_z~,
\end{equation}
where $1 \leq a \leq n $, $N(a)$ is the set of direct neighbours of photon $a$ along the state, and $k^{(a)}_\pm \in \{ 0, 1\}$, depending on the particular realisation of LC$_n$. The subscript on the operator $K$ denotes the state tensor-length of $K$, and hence the length of the state it acts upon. In fact we shall show the following statement

\begin{theorem}
The $n$-photon $S^{(n)}$ state satisfies the set of LC$_n$-eigenvalue equations for 

\begin{align}\label{params}
\begin{split}
k^{(a)}_+  &= 
\begin{cases}
1, & \text{if}\ a \in \{ 1, n\} \\
0, & \text{if}\ 1 < a < n
\end{cases} \\
\vspace{5mm}
k^{(a)}_-  &= 
\begin{cases}
1, & \text{if}\ a = 1 \\
0, & \text{if}\ 1 < a \leq n
\end{cases}
\end{split}
\end{align}

\end{theorem}

\begin{proof}
The proof follows, once again, by induction, as well as the use of Lemma~\ref{lemma1}

\vspace{5mm}
\textit{Basis case}: For $j=2$, .

\begin{align}
\begin{split}
S^{(2)}_+ = (\Ket{1_1}+\Ket{0_2}) \Ket{1_n} + (\Ket{1_1}-\Ket{0_1}) \Ket{0_2} ~, \\
S^{(2)}_+ = (\Ket{1_1}+\Ket{0_2}) \Ket{1_n} - (\Ket{1_1}-\Ket{0_1}) \Ket{0_2} ~, \\
\end{split}
\end{align}
and the statement holds when applying $\sigma^{(1)}_x \otimes \sigma^{(2)}_z$ and $\sigma^{(1)}_z \otimes \sigma^{(2)}_x$.

\vspace{5mm}
\textit{Induction step}: Suppose the statement holds for $j=n$, and consider $S^{(n+1)}_+ = S^{(n)}_+ \Ket{1_{n+1}} + S^{(n)}_- \Ket{0_{n+1}}$.Then
\vspace{5mm}

\textit{If $a=1$}:

\begin{align}
\begin{split}
K^{(a)}_{n+1} S^{(n+1)}_+ &= (K^{(a)}_{n} \otimes \mathbb{I}_2) (S^{(n)}_+ \Ket{1_{n+1}} + S^{(n)}_- \Ket{0_{n+1}}) \\
                                           &=  (-1)^{k^{(1)}_+} S^{(n)}_+ \Ket{1_{n+1}} + (-1)^{k^{(1)}_-} S^{(n)}_- \Ket{0_{n+1}} \\
                                           &= - (S^{(n)}_+ \Ket{1_{n+1}} + S^{(n)}_- \Ket{0_{n+1}}) \\
                                           &= - S^{(n+1)}_+~,
\end{split}
\end{align}
with $\mathbb{I}_2$ being the $2 \times 2$ identity matrix. The penultimate step holds due the induction hypothesis. Similarly, for $S^{(n+1)}_-$,

\begin{align}
\begin{split}
K^{(a)}_{n+1} S^{(n+1)}_- &= (K^{(a)}_{n} \otimes \mathbb{I}_2) (S^{(n)}_+ \Ket{1_{n+1}} - S^{(n)}_- \Ket{0_{n+1}}) \\
                                           &=  (-1)^{k^{(1)}_+} S^{(n)}_- \Ket{1_{n+1}} - (-1)^{k^{(1)}_-} S^{(n)}_- \Ket{0_{n+1}} \\
                                           &= - (S^{(n)}_+ \Ket{1_{n+1}} - S^{(n)}_- \Ket{0_{n+1}}) \\
                                           &= - S^{(n+1)}_-~.
\end{split}
\end{align}

\textit{If $1 < a < n$}:

\begin{align}
\begin{split}
K^{(a)}_{n+1} S^{(n+1)}_+ &= (K^{(a)}_{n} \otimes \mathbb{I}_2) (S^{(n)}_+ \Ket{1_{n+1}} + S^{(n)}_- \Ket{0_{n+1}}) \\
                                           &=  (-1)^{k^{(a)}_+} S^{(n)}_+ \Ket{1_{n+1}} + (-1)^{k^{(a)}_-} S^{(n)}_- \Ket{0_{n+1}} \\
                                           &= S^{(n)}_+ \Ket{1_{n+1}} + S^{(n)}_- \Ket{0_{n+1}} \\
                                           &= S^{(n+1)}_+~,
\end{split}
\end{align}

\begin{align}
\begin{split}
K^{(a)}_{n+1} S^{(n+1)}_- &= (K^{(a)}_{n} \otimes \mathbb{I}_2) (S^{(n)}_+ \Ket{1_{n+1}} - S^{(n)}_- \Ket{0_{n+1}}) \\
                                           &=  (-1)^{k^{(a)}_+} S^{(n)}_- \Ket{1_{n+1}} - (-1)^{k^{(a)}_-} S^{(n)}_- \Ket{0_{n+1}} \\
                                           &= S^{(n)}_+ \Ket{1_{n+1}} - S^{(n)}_- \Ket{0_{n+1}} \\
                                           &= S^{(n+1)}_-~.
\end{split}
\end{align}

\textit{If $a = n$}:

\begin{align}
\begin{split}
K^{(a)}_{n+1} S^{(n+1)}_+ &= (K^{(a)}_{n} \otimes \sigma^{(n+1)}_z) (S^{(n)}_+ \Ket{1_{n+1}} + S^{(n)}_- \Ket{0_{n+1}}) \\
                                           &=  -(-1)^{k^{(n)}_+} S^{(n)}_+ \Ket{1_{n+1}} + (-1)^{k^{(n)}_-} S^{(n)}_- \Ket{0_{n+1}} \\
                                           &= S^{(n)}_+ \Ket{1_{n+1}} + S^{(n)}_- \Ket{0_{n+1}} \\
                                           &= S^{(n+1)}_+~,
\end{split}
\end{align}

\begin{align}
\begin{split}
K^{(a)}_{n+1} S^{(n+1)}_- &= (K^{(a)}_{n} \otimes \sigma^{(n+1)}_z) (S^{(n)}_+ \Ket{1_{n+1}} - S^{(n)}_- \Ket{0_{n+1}}) \\
                                           &=  - (-1)^{k^{(n)}_+} S^{(n)}_- \Ket{1_{n+1}} - (-1)^{k^{(n)}_-} S^{(n)}_- \Ket{0_{n+1}} \\
                                           &= S^{(n)}_+ \Ket{1_{n+1}} - S^{(n)}_- \Ket{0_{n+1}} \\
                                           &= S^{(n+1)}_-~.
\end{split}
\end{align}

For the $a = n+1$ case, we shall make use of Lemma~\ref{lemma2}. The operator $K^{(n+1)}_{n+1}$ can be decomposed as $\mathbb{I}_2^{\bigotimes_{n-1}} \otimes \sigma^{(n)}_z \otimes \sigma^{(n+1)}_x$, and hence we get that

\vspace{5mm}
\textit{If $a = n+1$}:

\begin{align}
\begin{split}
K^{(a)}_{n+1} S^{(n+1)}_+ &= - S^{(n)}_- \Ket{0_{n+1}} - S^{(n)}_+ \Ket{1_{n+1}} \\
                                           &= - S^{(n+1)}_+~,
\end{split}
\end{align}

\begin{align}
\begin{split}
K^{(a)}_{n+1} S^{(n+1)}_-  &=  -S^{(n)}_- \Ket{0_{n+1}} + S^{(n)}_+ \Ket{1_{n+1}} \\
                                           &= S^{(n+1)}_-~.
\end{split}
\end{align}    
Therefore, the states $S^{(n)}_\pm$ satisfy the eigenvalue conditions \eqref{conditions} for the set of parameters \eqref{params}, meaning that the the state obtained by our protocol is an LC$_n$ state.
\end{proof}

\section{Average fidelity}
\label{app:fidel}

Consider a single scattering event in which the spin precesses for a time $T^{(1)}_B$ prior to the scattering event and a subsequent precession time $T^{(2)}_B$ followed by a Y rotation marking the end of the run (such that $T^{(1)}_B + T^{(2)}_B = T_B$). In the presence of the $B^x_N$ component, the rotation matrix $U_p(\phi)$ in \eqref{matrices} picks up a stochastic term $\omega_N t$, that is

\begin{equation}
U_p((\omega_B + \omega_N) t) = \left(\begin{array}{cc}\mathrm{e}^{-i \frac{1}{2}(\omega_B + \omega_N) t} & 0 \\0 & \mathrm{e}^{i \frac{1}{2}(\omega_B + \omega_N) t}\end{array}\right) \otimes \mathbb{I}_3^{\bigotimes_{n}}~,
\end{equation}
with $t = T^{(1)}_B$ or $T^{(2)}_B$, where we have written the precessed angle explicitly in terms of $\omega_B = g^x_h \mu_B B_{ext} / \hbar$ and the Overhauser stochastic frequency $\omega_N = g^x_h \mu_B B^x_N / \hbar$ ($g_h^x$ being the $x$ component of the anisotropic hole g-factor\footnote{The anisotropy in the hole g-factor is, in general, not the same as the effective anisotropy in the g-tensor for the hole Overhauser shift due to hlh--lh mixing.}). 

The effect of this stochastic term can be seen in the trace fidelity between post Y rotation ideal photon state, and the more realistic case including the Overhauser field. The spin+photon states for the two cases, denoted by $S^{(1)}$ and $\tilde{S}^{(1)}$, respectively, are then given by

\begin{align}
\begin{split}
S^{(1)}  = &\mathrm{e}^{-i \frac{1}{2}\omega_B  \delta T_B}\Ket{\Ua B_1} +\mathrm{e}^{-i \frac{1}{2}\omega_B  \delta T_B}\Ket{\Da B_1} \\
             + &\mathrm{e}^{i \frac{1}{2}\omega_B  \delta T_B}\Ket{\Ua R_1} - \mathrm{e}^{i \frac{1}{2}\omega_B  \delta T_B}\Ket{\Da R_1}~, \\
\vspace{5mm}\\
\tilde{S}^{(1)}  = &\mathrm{e}^{-i \frac{1}{2}(\omega_B + \omega_N) \delta T_B}\Ket{\Ua B_1} +\mathrm{e}^{-i \frac{1}{2}(\omega_B + \omega_N) \delta T_B}\Ket{\Da B_1} \\
                             + &\mathrm{e}^{i \frac{1}{2}(\omega_B + \omega_N) \delta T_B}\Ket{\Ua R_1} - \mathrm{e}^{i \frac{1}{2}(\omega_B + \omega_N) \delta T_B}\Ket{\Da R_1}~,
\end{split}
\end{align}
where $ \delta T_B = T^{(1)}_B - T^{(2)}_B \in [-T_B,T_B]$ is a uniform random variable due to the fact that the spin precesses multiple times during $T_B$ in the high external magnetic field. The final photon state, as discussed earlier, depends on the state the spin is measured in, so we shall denote the density matrices of the ideal and realistic cases by $\rho^{(1)}_+$ and $\xi^{(1)}_+$, respectively, if the spin is measured in the $\Ket{\Ua}$ state, and similarly $\rho^{(1)}_-$ and $\xi^{(1)}_-$ for the $\Ket{\Da}$ result. The fidelity for a fixed value of $B^x_N$ is then given by $ \mathcal{F}^{(1)} = \mathrm{tr}(\rho^{(1)}_+ \xi^{(1)}_+) = \mathrm{tr}(\rho^{(1)}_- \xi^{(1)}_-) = \cos^2( B^x_N \delta T_B / 2 )$. 

Due to the stochastic nature of the Overhauser field, we need to ensemble-average $\mathcal{F}^{(1)}$ in order to get the true fidelity, that is $\bar{\mathcal{F}}^{(1)} = \Braket{\Braket{\mathrm{tr}(\rho^{(1)}_- \xi^{(1)}_-)}_B}_{\delta T} = \Braket{\Braket{\mathrm{tr}(\rho^{(1)}_+ \xi^{(1)}_+)}_B}_{\delta T}$, where the Overhauser averaging $\Braket{\cdot}_B$ and time averaging $\Braket{\cdot}_{\delta T}$ are performed over a normal distribution with zero mean and finite standard deviation $\delta B^x_N$, and a uniform distribution over $[-T_B,T_B]$ \footnote{the averages are performed independently due to the statistical independence of $B^x_N$ and $\delta T$}. In doing so, we get the averaged fidelity for a single scattering event in the presence of $B^x_N$ given by Eqn.~\ref{fidelav}.

\section{Imperfections of other QD-based protocols}
\label{app:compare}

As discussed in the main text, several protocols have been proposed for implementing photonic LC states or entangled states sharing similar properties. The influential 2009 proposal by Lindner and Rudolph\cite{Terry2009} (LR) offered an elegant and simple scheme which could be implemented using the circularly polarised degrees of freedom of a quantum dot. Despite its simplicity, a number of experimental barriers need to be overcome to actually implement such a scheme. The Overhauser fluctuation limitations have already been discussed in the main text; below we discuss some additional constraints both for the LR scheme as well as the recent dark exciton (DE) based LC scheme\cite{Schwartz2016}, which has already successfully produced LC$_2$ states in the laboratory and shown promise for reaching up to LC$_5$. In essence, these imperfections effectively introduce limits to the size of achievable cluster states for those protocols, hence limiting the indefinite deterministic operation in the absence of further optimisations. By contrast, we note that our approach in this work -- as discussed in the main paper -- is largely immune against all issues discussed below.

\subsection{Shortcomings due to coupling to phonons}

The solid-state environment further limits the deterministic nature of these protocols due to coupling to the phonon environment. Even in the limit of idealised instantaneous excitation pulses, a temperature-dependent fraction of the photons are inevitably emitted incoherently via the phonon sideband ($\sim 9\%$ at temperatures as low as $T=4$K, increasing with temperature \cite{Smith2017}). This affects all protocols involving electronic excitation to trion of biexciton states, i.e. both the LR and DE approaches.

\subsection{Effects of hole state mixing}
\label{app:mix}

In this section, we discuss how said protocols fare against finite hh-lh mixing\footnote{Admixture of conduction band states even in the absence of a lh contribution may result in a non-Ising type hyperfine Hamiltonian for the hh system \cite{Prechtel2016}, however, this goes beyond the scope of this work}. The first type of hh-lh mixing, due to anisotropy in the in-plane strain of the quantum dot, gives rise to the hh$\ua$-lh$\da$ mixing, resulting in the hole eigenstates 
\begin{align}\label{ud}
\begin{split}
\Ket{\Ua} = \frac{1}{\sqrt{1+|\beta_{ud}|^2}}(\Ket{3/2;+3/2} + \beta_{ud}\Ket{3/2;-1/2}) ~, \\
\Ket{\Da} = \frac{1}{\sqrt{1+|\beta_{ud}|^2}}(\Ket{3/2;-3/2} + \beta^*_{ud}\Ket{3/2;+1/2})~,
\end{split}
\end{align}
where, without giving its explicit form, $\beta_{ud}$ is the in-plane strain-dependent mixing factor \cite{Testelin2009, Prechtel2016}. This type of mixing primarily causes ellipticity of the dipole-allowed transitions which, for a hh system, would be driven by $\sigma^\pm$ polarized light. Hence this hh$\ua$-lh$\da$ mixing does not induce the `diagonal' dipole-forbidden transitions. 

On the other hand, the hh$\ua$-lh$\ua$ mixing may allow transitions which would otherwise be forbidden for a hh system. The hole eigenstates solely due this type of mixing are given by:
\begin{align}\label{uu}
\begin{split}
\Ket{\Ua} = \frac{1}{\sqrt{1+|\beta_{uu}|^2}}(\Ket{3/2;+3/2} + \beta_{uu}\Ket{3/2;+1/2}) ~, \\
\Ket{\Da} = \frac{1}{\sqrt{1+|\beta_{uu}|^2}}(\Ket{3/2;-3/2} + \beta^*_{uu}\Ket{3/2;-1/2})~,
\end{split}
\end{align}
where $\beta_{uu}$ is the hh$\ua$-lh$\ua$ admixture factor \cite{Testelin2009, Prechtel2016}. From Eqns.~\eqref{uu}, it can be immediately seen that the transitions, which are forbidden in Faraday geometry, are now allowed. For hole-spins, $\beta_{uu}$ has been measured to be $\sim 8\%$, leading to allowed-to-forbidden transition ratios of $|\beta_{uu}|^2 / 3 \approx 0.2\%$ \cite{Prechtel2016}, although this varies from one quantum dot to another. This means that even if the external field in the LR scheme is weak enough to preserve a pure Faraday geometry, dipole-forbidden transitions may still occur with some small, but finite probability, both for the original and the hole-spin variant of the LR protocol.

Similarly, in the DE system $z$-polarised `forbidden' transitions are also present due to hole sub-band mixing, although these transitions in this system are significantly weaker \cite{GershoniAtomistic, Dupertuis2011}. In addition to hh-lh mixing, the DE scheme also suffers from dark-bright exciton (DE-BE) state mixing due to the breaking of the C$_{2v}$ symmetry, although this effect is much weaker than the hh-lh mixing. Realistically, self-assembled QDs suffer from a reduction in symmetry during the growth process, causing a departure from the ideal C$_{2v}$ symmetry. The resulting `reduced' C$_s$ symmetry leads to DE-BE state couplings of two kinds; the first leads to finite $z$-polarised dipole transitions similar to the hh$\ua$-lh$\da$ admixture in the BE schemes, whilst the second gives rise to forbidden in-plane transitions, bearing similar repercussions as the hh$\ua$-lh$\ua$ mixing discussed above \cite{GershoniAtomistic, GershoniPhenom}, although to a much lesser extent.

We note that our approach does not suffer from modifications of the selections rules due to hole mixing: we already rely on the presence of off-diagonal transitions and slight changes to their rates will not make an appreciable difference.
\begin{figure}[t!]
\centering
\includegraphics[width=.8\linewidth]{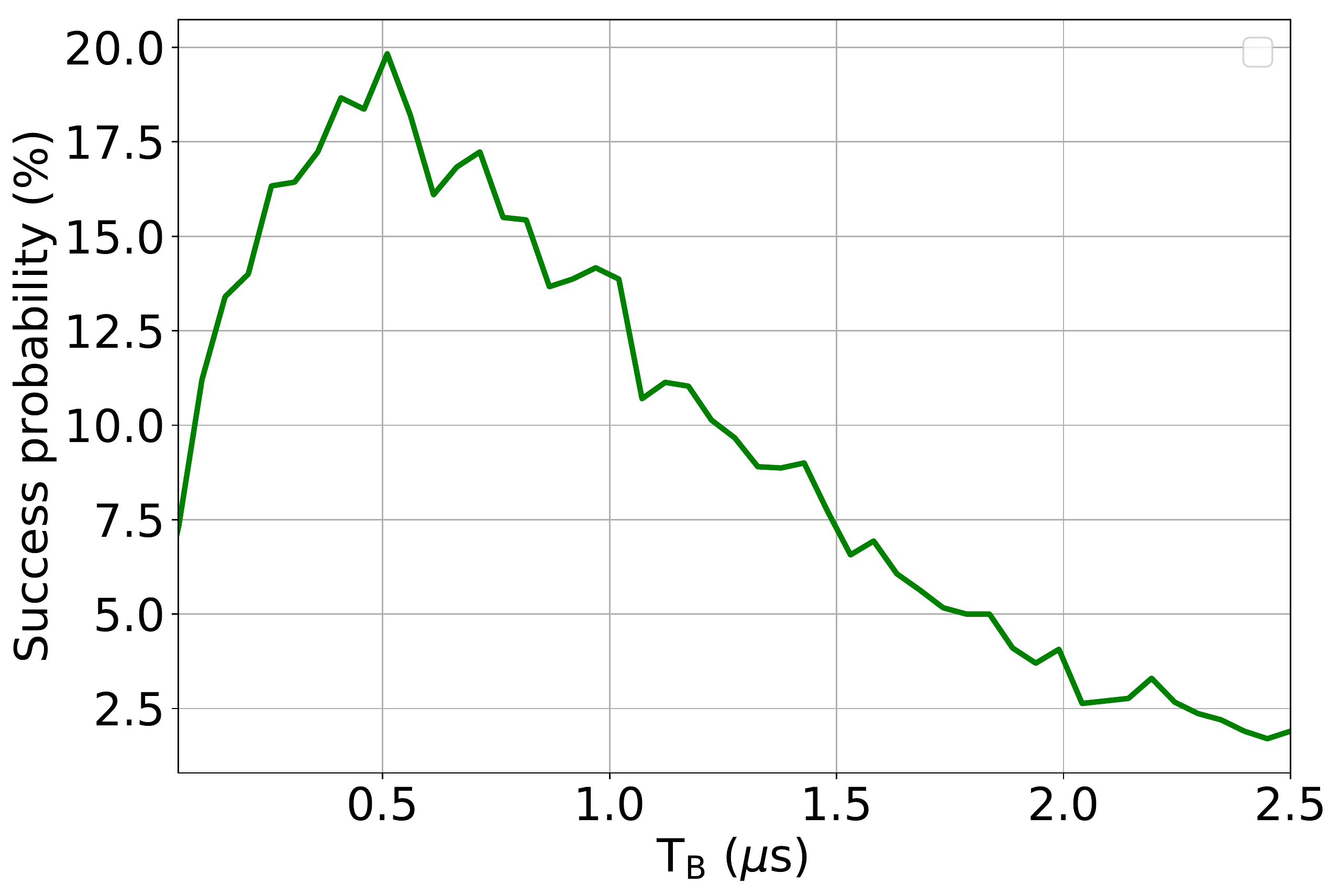} 
\caption{Success probability of scattering a single Raman photon using a c.w. source against time-bin length $T_B$. At an optimised time-bin length $T_B \approx 0.5\mu$s, the probability can be as high as 20\%, before it drops once more due to the probabilities of getting multiple photons in a single time-bin.}
\label{prob}
\end{figure}

\subsection{Pulsed scheme limitations}

As mentioned earlier, the main limitation of our scheme is the unknown time-of-arrival of the photons due to the c.w. source. An obvious solution might be using a pulsed source for the photons. Despite addressing the issue of the photons' unknown phases, such a protocol would still not be deterministic, as there is still a 50\% probability that a Rayleigh scattering event, instead of a spin-flipping Raman one, occurs. Whilst this is still a considerable improvement over the $\sim$ 20\% we get for an optimised time-bin length (Fig.~\ref{prob}), this pulsed-excitation scheme would not benefit from the advantages of sub-saturation driving; mainly the photon linewidth limited only by the hole spin coherence and laser linewidth, and be susceptible to phonon dephasing. Hence, the opportunity to create longer LCs with less probabilistic phase uncertainty comes at the price of lower quality LC states, which we argue is paramount for reliably constructing 2D cluster states required for quantum computation using probabilistic fusion gates.

\section{Robustness of 2D cluster state protocols and LC state fusing schemes}
\label{app:2d}
\begin{figure}[t!]
\flushleft
\includegraphics[width=1.01\linewidth]{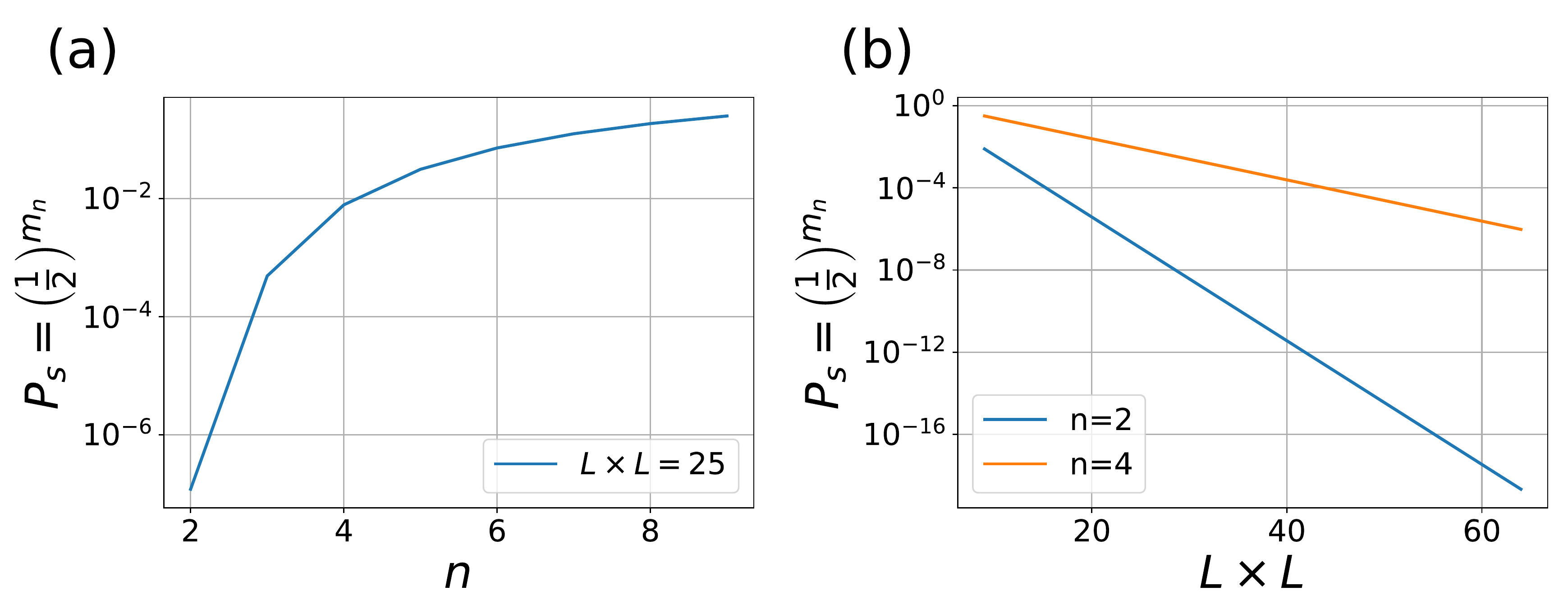} 
\caption{\textbf{a)} Success probability $P_s$ of obtaining a $5 \times 5$ 2D cluster state as a function of the length $n$ of the input LC states to be fused. Going from LC$_2$ to LC$_4$ shows orders of magnitude improvement, underlining that having at least moderately sized LC states is essential for feasible 2D state growth. \textbf{b)} Success probability $P_s$ against 2D cluster state size $L \times L$ for LC$_2$ (bottom, blue line) and LC$_4$ (top, orange line) `building blocks', showing an increased improvement with size when going from one-dimensional states of size 2 to size 4.}
\label{growth}
\end{figure}

Schemes extending the LR scheme for 2D cluster state generation have been proposed \cite{Economou2010}, in which it was shown that a pair of entangled QDs could be used to directly generate a 2D cluster state, reducing the required number of probabilistic fusion of LC state building blocks. Furthermore, it was recently shown that the requirement of two-qubit gates on the entangling emitters can be relaxed by a careful application of pulses and single-qubit gates on the emitters \cite{Terry2018}. However, building on a similar setup and selection rules as the original LR protocol, we expect that the practical limitations discussed above will also limit the achievable size of photonic states that can be obtained with this protocol.

An alternative approach to generating a 2D cluster state is that of fusing LC states. We show that having high fidelity LC states of moderate length is essential for using one-dimensional states as building blocks. Consider a 2D cluster state of size $L \times L$. If we start with number of linear cluster states of size $n$, then the number of steps required to at least reach a 2D cluster state of size $L \times L$ is at least $m_n = \frac{L^2 - n}{n-1}$: assuming that we have enough linear clusters to start with, each fusion process will (on average) increases the cluster size by $n(m_n+1)-m_n$ (noting that each fusion step leaves the fused qubit redundantly encoded with 2 photons in type II fusion, and disregarding the final layout of the 2D state for generality and simplicity). Clearly, we ignore the cases when $n > L^2$ as the probability saturates for $n=L^2$. We show how the probability scales for a 2D cluster state of size $5 \times 5$ as a function of the `building block' size (i.e. the size of the initial cluster states) in Fig.~\ref{growth}a). This clearly demonstrates that the probability increases exponentially before saturating, showing a significant jump when going from linear cluster sizes of 2 to 4.

This increase in success probability is further emphasised when one considers increasing the 2D cluster state size. In Fig.~\ref{growth}b), we show how the difference in probability increases with increasing 2D state size $L \times L$. This approach assumes that upon failure, we have enough resources to replace the linear cluster state and try again. The results of this relatively na\"ive and basic analysis are further backed by an alternative approach presented in Ref.~\citen{Gross2006}, in which Gross et al. fuse linear clusters by `weaving' $n+1$ linear clusters of size $n$ to form a cluster state of size $n \times n$. They show that as long as a careful choice of parameters is made, depending on the fusion success probability, then the cluster state can be prepared using $O(n^2)$ edges and the overall success probability approaches unity as n goes to infinity.

Besides having relatively longer linear states as building blocks, the fidelity of these states, indicative of quality, is also an important factor when considering scalability to higher dimensions \cite{Segovia2015, Morley2017}, as it will determine the `percolation' or `edge-bound' probability. Fortunately, our approach can deliver on both counts by producing LC$_4$ states with high fidelity at a respectable generation rate.

\section{Proposal for deterministic scheme using DQD}
\label{app:Deter}

Motivated by recent theoretical and experimental work, we propose extending our Raman protocol to a double quantum dot (DQD) system, where, depending on the relative strength of the exchange interaction and transition energy detuning between the two QDs, either joint measurements on the DQD system can be performed, whilst leaving the photon-entangling hole spin state unaffected, or oscillations between joint states can be detected without collapsing the system joint state.  In the following, we will discuss two possibilities of extending our protocol in such a way.

A)\textit{Electrical control}:
During the past few years, great progress has been made in synthesising and controlling quantum dot molecules, both in stacked \cite{Stinaff2006, Krenner2005, Kim2011} and lateral \cite{Songmuang2003, Wang2008} geometries. A Raman-spin flip DQD scheme was shown in Ref.~\citen{Vora2015}, in which the external field is applied in Faraday geometry and the Raman spin-flips occur between the singlet $S$ and triplet $T_0$ states of the system. Whilst this configuration would not allow screening of the dominant fluctuation component of the Overhauser field, such a setup would, in principle, allow a current measurement scheme to be applied and signal the Raman events. In fact, the standard singlet-triplet spin-blockade used in gated-DQDs \cite{Shaji2008} could be used to detect current drops, signalling the Raman event. This would require operation round the (1,0),(1,1),(2,0) triple-point at a negative bias, making use of the the additional charge state $S(2,0)$.
\begin{figure}[t!]
\centering
\includegraphics[width=\linewidth]{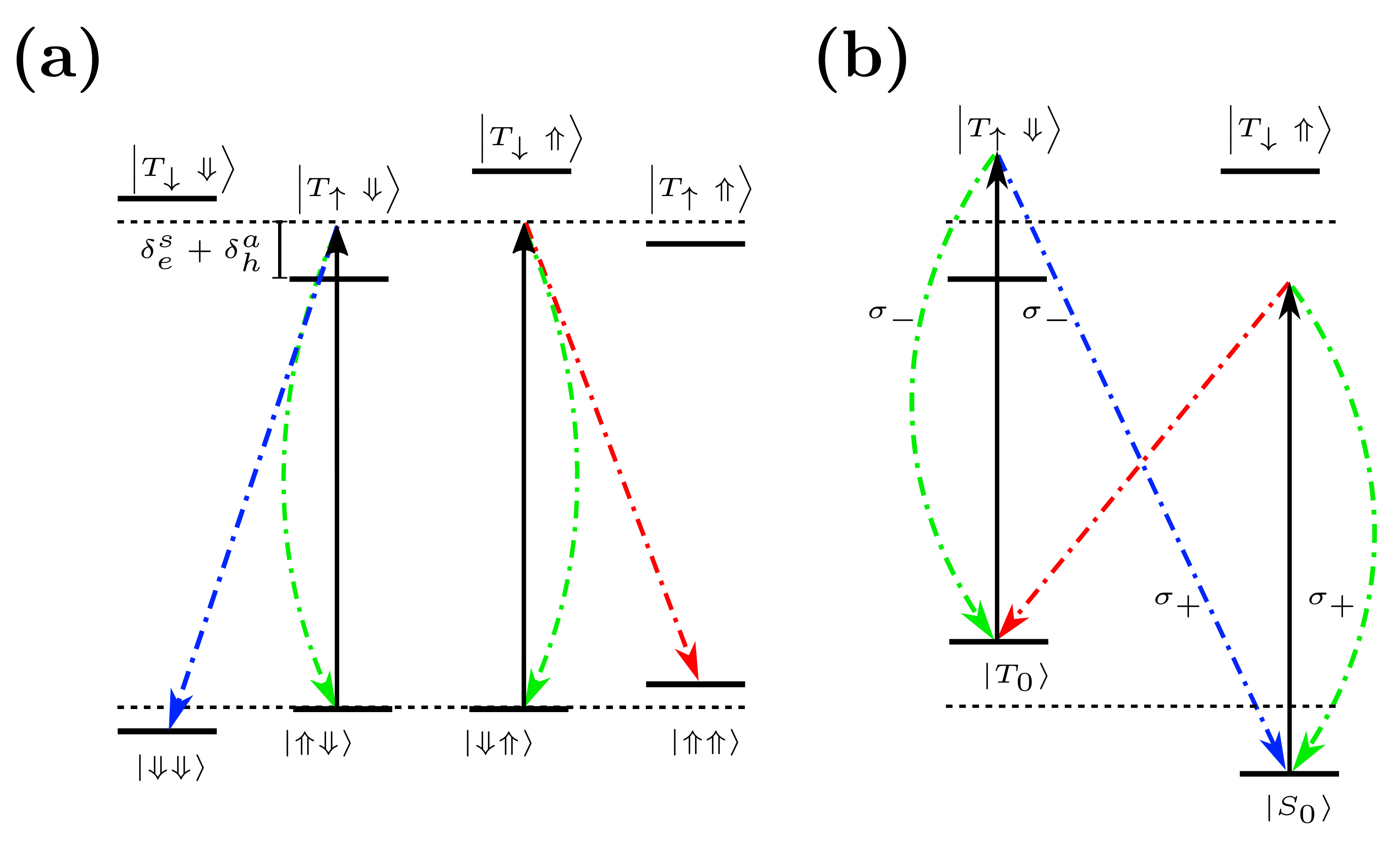} 
\caption{\textbf{a)} Extending the Raman spin-flip protocol to a DQD setup in Voigt geometry, where the two QDs are sufficiently detuned (relative to the exchange interaction), allowing the optical addressing of a single spin. \textbf{b)} An alternative setup in Faraday geometry \cite{Vora2015}, in which the initial state would be a superposition of $S$ and $T_0$ states.}
\label{DQD}
\end{figure}
Addressing and manipulation of these singlet and triplet states in optically-active DQDs have been recently been demonstrated for QD molecules \cite{Vora2015, Greilich2011, Kim2011, Weiss2012, Elzerman2011}, whereas the current transport measurements has been long understood for surface-defined QDs. This route would require a hybrid gated and optically-active device, which, although certainly challenging, might nonetheless present a feasible route.

B)\textit{Optical control}: 
A more attractive alternative to having a gated structure would be to have an all-optical non-invasive spin readout technique, provided by the rich energy level structure for these systems. In quantum dot molecules, this can be achieved by using the distributed trion state, with the ancilla spin being empty, whilst the host spin being singly-electron charged. The spin readout technique was demonstrated experimentally performing resonance fluorescence (RF) on the $\Ket{\downarrow_s, 0_a} \leftrightarrow \Ket{\downarrow_s,\downarrow \Uparrow_a}$ transition, which is decoupled from the main spin-flip transition \cite{Vamivakas2010}. This technique could be readily extended to hole spin systems with an analogous level structure. A similar setup was demonstrated experimentally in Ref.~\citen{Kim2008}, where use of these cycling transitions was made to detect the flips of the host spin state.  Both these setups would require individual addressing of the ancilla and host spin, meaning that the two QDs selected must be sufficiently relatively far-detuned, which could be achieved by tuning the bias voltage over the sample, decreasing the exchange energy splitting \cite{Greilich2011}. Alternatively, for samples with a much stronger singlet-triplet splitting, optical addressing of the joint states would be more feasible. In the singlet-triplet Raman scheme in Faraday geometry discussed  in Ref.~\citen{Vora2015} (Fig.~\ref{DQD}b), spin readout of the singlet state can be performed by using the decoupled cycling transition $T_+ \leftrightarrow R_{++}$ \cite{Weiss2012}.

\section{Quantum state tomography}
\label{app:Tom}

\begin{figure}[t!]
\flushleft
\includegraphics[width=1\linewidth]{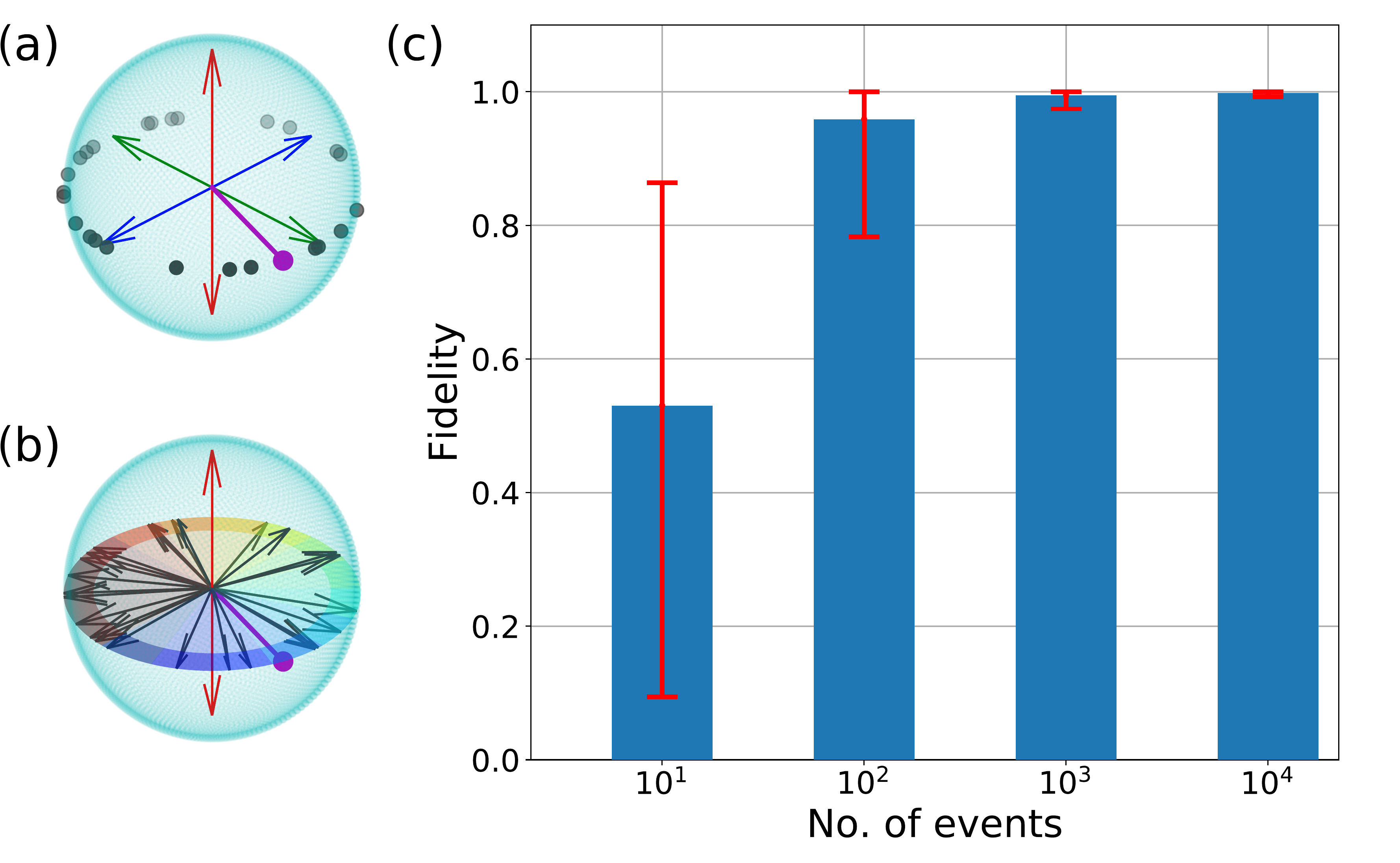} 
\caption{\textbf{a)} Bloch sphere representation of the problem: the actual state  to be reconstructed (purple) gains a random phase (dots) prior to every measurement, with the measurement bases given by the arrows. \textbf{b)} An equivalent picture where the state is fixed, with the `random' measurement basis given by the phases. \textbf{c)} Fidelity for 16 grouped projectors, showing, as expected, an increase in the Fidelity for higher numbers of events.}
\label{tomo}
\end{figure}
Quantum state tomography (QST) allows one to completely characterise an unknown quantum state, as long as an ensemble of identical copies of such a state can be created in the experiment. Despite the wide range of tomographic techniques in existence \cite{Bolduc2017, Kwiat2001,Paris2004}, the aim is typically to use sets of repeated measurements on the ensemble, the results of which enable the reconstruction of the original state.

Our probabilistic protocol then presents an obvious question as to how would one obtain multiple copies of the cluster state, since for each realisation, the phases imprinted on the photonic qubits are random, and are only know post-detection. Using conventional reconstruction techniques would then average over the coherences of the cluster state, losing the entanglement information.

Despite this apparent downfall, the fact that the random phases can be determined post-measurement means that this problem can be reformulated in a `static frame' with respect to the state, that is, the state is not imprinted with the phases, instead, in this frame, the effective basis chosen for the actual measurement rotates for each measurement due to the different phases. The problem then reduces to reconstructing a state when the measurement basis used is different for each measurement. We emphasise that this does not mean that the experimentally chosen basis is actually rotated for each measurement. The random measurement projectors can then be grouped as $\{P^{(j)}_1,P^{(j)}_2,...P^{(j)}_{n_j}\}$ by proximity on the Bloch sphere into projectors $\{\mathcal{P}_j:~ 1 \leq j \leq K\}$ to be used for reconstruction. 

As a proof of principle, we used the state $c_1 \Ket{H} + c_2 \Ket{V}$ for each experiment, where $c_1$ and $c_2$ are two random complex numbers, so that each time we have a different measurement projector. Using Maximum Likelihood Estimation (MLE) and the Cholesky decomposition for the density matrix, we performed QST for various numbers of grouped projectors, the results of which are shown in Fig.~\ref{tomo}. As the number of of $\mathcal{P}_j$ is increased, the fidelity rises, as expected. However, for lower numbers of events, the fidelity peaks at a number of grouped projectors, and then starts declining again. This drop is due to the failure of the Gaussian assumption used for MLE. This failure is expected to affect fidelities for higher event numbers as we increase the number of projectors.

\bibliography{RamanProtBib}

\begin{thebibliography}{10}

\bibitem{Briegel2001b}
Hans~J. Briegel and Robert Raussendorf.
\newblock Persistent entanglement in arrays of interacting particles.
\newblock {\em Phys. Rev. Lett.}, 86:910--913, Jan 2001.

\bibitem{Briegel2001a}
Robert Raussendorf and Hans~J. Briegel.
\newblock A one-way quantum computer.
\newblock {\em Phys. Rev. Lett.}, 86:5188--5191, May 2001.

\bibitem{Raussendorf2007}
R.~Raussendorf, J.~Harrington, and K.~Goyal.
\newblock Topological fault-tolerance in cluster state quantum computation.
\newblock {\em New Journal of Physics}, 9(6):199, 2007.

\bibitem{OBrien2009}
Jeremy~L. O'Brien, Akira Furusawa, and Jelena Vuckovic.
\newblock Photonic quantum technologies.
\newblock {\em Nat Photon}, 3(12):687--695, Dec 2009.

\bibitem{TerryFusion2005}
Daniel~E. Browne and Terry Rudolph.
\newblock Resource-efficient linear optical quantum computation.
\newblock {\em Phys. Rev. Lett.}, 95:010501, Jun 2005.

\bibitem{Herrera2010}
David~A. Herrera-Mart\'{\i}, Austin~G. Fowler, David Jennings, and Terry
  Rudolph.
\newblock Photonic implementation for the topological cluster-state quantum
  computer.
\newblock {\em Phys. Rev. A}, 82:032332, Sep 2010.

\bibitem{Weinstein2011}
Y.S. Weinstein.
\newblock Fusing imperfect photonic cluster states.
\newblock {\em Journal of Modern Optics}, 58(14):1285--1291, 2011.

\bibitem{Nielsen2004}
Michael~A. Nielsen.
\newblock Optical quantum computation using cluster states.
\newblock {\em Phys. Rev. Lett.}, 93:040503, Jul 2004.

\bibitem{Terry2009}
Netanel~H. Lindner and Terry Rudolph.
\newblock Proposal for pulsed on-demand sources of photonic cluster state
  strings.
\newblock {\em Phys. Rev. Lett.}, 103:113602, Sep 2009.

\bibitem{Denning2017}
Emil~V. Denning, Jake Iles-Smith, Dara P.~S. McCutcheon, and Jesper M{\o}rk.
\newblock Protocol for generating multiphoton entangled states from quantum
  dots in the presence of nuclear spin fluctuations.
\newblock {\em Phys. Rev. A}, 96:062329, Dec 2017.

\bibitem{Barrett2005}
Sean~D. Barrett and Pieter Kok.
\newblock Efficient high-fidelity quantum computation using matter qubits and
  linear optics.
\newblock {\em Phys. Rev. A}, 71:060310, Jun 2005.

\bibitem{Lin2008}
Zhi-Rong Lin, Guo-Ping Guo, Tao Tu, Fei-Yun Zhu, and Guang-Can Guo.
\newblock Generation of quantum-dot cluster states with a superconducting
  transmission line resonator.
\newblock {\em Phys. Rev. Lett.}, 101:230501, Dec 2008.

\bibitem{Schwartz2016}
I.~Schwartz, D.~Cogan, E.~R. Schmidgall, Y.~Don, L.~Gantz, O.~Kenneth, N.~H.
  Lindner, and D.~Gershoni.
\newblock Deterministic generation of a cluster state of entangled photons.
\newblock {\em Science}, 2016.

\bibitem{Economou2016}
Sophia~E Economou and Pratibha Dev.
\newblock Spin-photon entanglement interfaces in silicon carbide defect
  centers.
\newblock {\em Nanotechnology}, 27(50):504001, 2016.

\bibitem{Vallone2007}
Giuseppe Vallone, Enrico Pomarico, Paolo Mataloni, Francesco De~Martini, and
  Vincenzo Berardi.
\newblock Realization and characterization of a two-photon four-qubit linear
  cluster state.
\newblock {\em Phys. Rev. Lett.}, 98:180502, May 2007.

\bibitem{Zou2005}
XuBo Zou and W.~Mathis.
\newblock Generating a four-photon polarization-entangled cluster state.
\newblock {\em Phys. Rev. A}, 71:032308, Mar 2005.

\bibitem{Economou2010}
Sophia~E. Economou, Netanel Lindner, and Terry Rudolph.
\newblock Optically generated 2-dimensional photonic cluster state from coupled
  quantum dots.
\newblock {\em Phys. Rev. Lett.}, 105:093601, Aug 2010.

\bibitem{Terry2018}
M.~Gimeno-Segovia, T.~Rudolph, and S.~E. Economou.
\newblock {Deterministic generation of large-scale entangled photonic cluster
  state from interacting solid state emitters}.
\newblock arXiv:1801.02599, January 2018.

\bibitem{Iles-Smith2017}
Jake Iles-Smith, Dara P.~S. McCutcheon, Ahsan Nazir, and Jesper M{\o}rk.
\newblock Phonon scattering inhibits simultaneous near-unity efficiency and
  indistinguishability in semiconductor single-photon sources.
\newblock {\em Nature Photonics}, 11:521, Jul 2017.
\newblock Article.

\bibitem{Loss2002}
Alexander~V. Khaetskii, Daniel Loss, and Leonid Glazman.
\newblock Electron spin decoherence in quantum dots due to interaction with
  nuclei.
\newblock {\em Phys. Rev. Lett.}, 88:186802, Apr 2002.

\bibitem{Testelin2009}
C.~Testelin, F.~Bernardot, B.~Eble, and M.~Chamarro.
\newblock Hole-spin dephasing time associated with hyperfine interaction in
  quantum dots.
\newblock {\em Phys. Rev. B}, 79:195440, May 2009.

\bibitem{Hansom2014}
Jack Hansom, Carsten H.~H. Schulte, Claire Le~Gall, Clemens Matthiesen, Edmund
  Clarke, Maxime Hugues, Jacob~M. Taylor, and Mete Atat{\"u}re.
\newblock Environment-assisted quantum control of a solid-state spin via
  coherent dark states.
\newblock {\em Nature Physics}, 10:725, Sep 2014.

\bibitem{Malein2016}
R.~N.~E. Malein, T.~S. Santana, J.~M. Zajac, A.~C. Dada, E.~M. Gauger, P.~M.
  Petroff, J.~Y. Lim, J.~D. Song, and B.~D. Gerardot.
\newblock Screening nuclear field fluctuations in quantum dots for
  indistinguishable photon generation.
\newblock {\em Phys. Rev. Lett.}, 116:257401, Jun 2016.

\bibitem{Merkulov2002}
I.~A. Merkulov, Al.~L. Efros, and M.~Rosen.
\newblock Electron spin relaxation by nuclei in semiconductor quantum dots.
\newblock {\em Phys. Rev. B}, 65:205309, Apr 2002.

\bibitem{Braun2002}
P.F. Braun, X.~Marie, L.~Lombez, B.~Urbaszek, T.~Amand, P.~Renucci, V.~K.
  Kalevich, K.~V. Kavokin, O.~Krebs, P.~Voisin, and Y.~Masumoto.
\newblock Direct observation of the electron spin relaxation induced by nuclei
  in quantum dots.
\newblock {\em Phys. Rev. Lett.}, 94:116601, Mar 2005.

\bibitem{Chekhovich2013}
EA~Chekhovich, MN~Makhonin, AI~Tartakovskii, Amir Yacoby, H~Bluhm, KC~Nowack,
  and LMK Vandersypen.
\newblock Nuclear spin effects in semiconductor quantum dots.
\newblock {\em Nature materials}, 12(6):494, 2013.

\bibitem{Urbaszek2013}
Bernhard Urbaszek, Xavier Marie, Thierry Amand, Olivier Krebs, Paul Voisin,
  Patrick Maletinsky, Alexander H\"ogele, and Atac Imamo\ifmmode~\breve{g}\else
  \u{g}\fi{}lu.
\newblock Nuclear spin physics in quantum dots: An optical investigation.
\newblock {\em Rev. Mod. Phys.}, 85:79--133, Jan 2013.

\bibitem{Viola1999}
Lorenza Viola, Emanuel Knill, and Seth Lloyd.
\newblock Dynamical decoupling of open quantum systems.
\newblock {\em Phys. Rev. Lett.}, 82:2417--2421, Mar 1999.

\bibitem{Witzel2007}
W.~M. Witzel and S.~Das~Sarma.
\newblock Concatenated dynamical decoupling in a solid-state spin bath.
\newblock {\em Phys. Rev. B}, 76:241303, Dec 2007.

\bibitem{Zhang2007}
Wenxian Zhang, V.~V. Dobrovitski, Lea~F. Santos, Lorenza Viola, and B.~N.
  Harmon.
\newblock Dynamical control of electron spin coherence in a quantum dot: A
  theoretical study.
\newblock {\em Phys. Rev. B}, 75:201302, May 2007.

\bibitem{Uhrig2007}
G\"otz~S. Uhrig.
\newblock Keeping a quantum bit alive by optimized $\ensuremath{\pi}$-pulse
  sequences.
\newblock {\em Phys. Rev. Lett.}, 98:100504, Mar 2007.

\bibitem{Uhrig2008}
G\"otz~S. Uhrig.
\newblock Exact results on dynamical decoupling by π pulses in quantum
  information processes.
\newblock {\em New Journal of Physics}, 10(8):083024, 2008.

\bibitem{Bluhm2010}
Hendrik Bluhm, Sandra Foletti, Izhar Neder, Mark Rudner, Diana Mahalu, Vladimir
  Umansky, and Amir Yacoby.
\newblock Dephasing time of {GaAs} electron-spin qubits coupled to a nuclear
  bath exceeding 200$\mu$s.
\newblock {\em Nature Physics}, 7:109, Dec 2010.

\bibitem{Stockill2016}
R.~Stockill, C.~Le~Gall, C.~Matthiesen, L.~Huthmacher, E.~Clarke, M.~Hugues,
  and M.~Atat\"ure.
\newblock Quantum dot spin coherence governed by a strained nuclear
  environment.
\newblock {\em Nature Communications}, 7:12745, Sep 2016.
\newblock Article.

\bibitem{Eble2006}
B.~Eble, O.~Krebs, A.~Lema\^{\i}tre, K.~Kowalik, A.~Kudelski, P.~Voisin,
  B.~Urbaszek, X.~Marie, and T.~Amand.
\newblock Dynamic nuclear polarization of a single charge-tunable
  $\mathrm{InAs}/\mathrm{GaAs}$ quantum dot.
\newblock {\em Phys. Rev. B}, 74:081306, Aug 2006.

\bibitem{Petta2008}
J.~R. Petta, J.~M. Taylor, A.~C. Johnson, A.~Yacoby, M.~D. Lukin, C.~M. Marcus,
  M.~P. Hanson, and A.~C. Gossard.
\newblock Dynamic nuclear polarization with single electron spins.
\newblock {\em Phys. Rev. Lett.}, 100:067601, Feb 2008.

\bibitem{Majcher2017}
G.~\'Ethier-Majcher, D.~Gangloff, R.~Stockill, E.~Clarke, M.~Hugues,
  C.~Le~Gall, and M.~Atat\"ure.
\newblock Improving a solid-state qubit through an engineered mesoscopic
  environment.
\newblock {\em Phys. Rev. Lett.}, 119:130503, Sep 2017.

\bibitem{Note1}
The small Rabi energies entails that any dephasing due to the optical AC-Stark
  shift is negligible, although in principle AC-Stark shift tuning could be
  employed to significantly reduce the dephasing due to charge noise \cite
  {Ramsay2016}.

\bibitem{Becker2016}
Jonas~Nils Becker, Johannes G{\"o}rlitz, Carsten Arend, Matthew Markham, and
  Christoph Becher.
\newblock Ultrafast all-optical coherent control of single silicon vacancy
  colour centres in diamond.
\newblock {\em Nature Communications}, 7:13512, Nov 2016.
\newblock Article.

\bibitem{Yale2013}
Christopher~G. Yale, Bob~B. Buckley, David~J. Christle, Guido Burkard,
  F.~Joseph Heremans, Lee~C. Bassett, and David~D. Awschalom.
\newblock All-optical control of a solid-state spin using coherent dark states.
\newblock {\em Proceedings of the National Academy of Sciences},
  110(19):7595--7600, 2013.

\bibitem{Novikov2015}
S.~Novikov, T.~Sweeney, J.~E. Robinson, S.~P. Premaratne, B.~Suri, F.~C.
  Wellstood, and B.~S. Palmer.
\newblock Raman coherence in a circuit quantum electrodynamics lambda system.
\newblock {\em Nature Physics}, 12:75, Nov 2015.
\newblock Article.

\bibitem{Liu2016}
Qi-Chun Liu, Tie-Fu Li, Xiao-Qing Luo, Hu~Zhao, Wei Xiong, Ying-Shan Zhang,
  Zhen Chen, J.~S. Liu, Wei Chen, Franco Nori, J.~S. Tsai, and J.~Q. You.
\newblock Method for identifying electromagnetically induced transparency in a
  tunable circuit quantum electrodynamics system.
\newblock {\em Phys. Rev. A}, 93:053838, May 2016.

\bibitem{Premaratne2017}
Shavindra~P. Premaratne, F.~C. Wellstood, and B.~S. Palmer.
\newblock Characterization of coherent population-trapped states in a
  circuit-{QED} $\mathrm{\ensuremath{\Lambda}}$ system.
\newblock {\em Phys. Rev. A}, 96:043858, Oct 2017.

\bibitem{Loss1998}
Daniel Loss and David~P. DiVincenzo.
\newblock Quantum computation with quantum dots.
\newblock {\em Phys. Rev. A}, 57:120--126, Jan 1998.

\bibitem{Fischer2008}
Jan Fischer, W.~A. Coish, D.~V. Bulaev, and Daniel Loss.
\newblock Spin decoherence of a heavy hole coupled to nuclear spins in a
  quantum dot.
\newblock {\em Phys. Rev. B}, 78:155329, Oct 2008.

\bibitem{Chekhovich2011}
E.~A. Chekhovich, A.~B. Krysa, M.~S. Skolnick, and A.~I. Tartakovskii.
\newblock Direct measurement of the hole-nuclear spin interaction in single
  $\mathrm{InP}/\mathrm{GaInP}$ quantum dots using photoluminescence
  spectroscopy.
\newblock {\em Phys. Rev. Lett.}, 106:027402, Jan 2011.

\bibitem{Fallahi2010}
P.~Fallahi, S.~T. Y{\i}lmaz, and A.~Imamo\ifmmode~\breve{g}\else \u{g}\fi{}lu.
\newblock Measurement of a heavy-hole hyperfine interaction in {InGaAs} quantum
  dots using resonance fluorescence.
\newblock {\em Phys. Rev. Lett.}, 105:257402, Dec 2010.

\bibitem{Bulaev2005}
Denis~V. Bulaev and Daniel Loss.
\newblock Spin relaxation and decoherence of holes in quantum dots.
\newblock {\em Phys. Rev. Lett.}, 95:076805, Aug 2005.

\bibitem{DeGreve2011}
Kristiaan De~Greve, Peter~L. McMahon, David Press, Thaddeus~D. Ladd, Dirk
  Bisping, Christian Schneider, Martin Kamp, Lukas Worschech, Sven H{\"o}fling,
  Alfred Forchel, and Yoshihisa Yamamoto.
\newblock Ultrafast coherent control and suppressed nuclear feedback of a
  single quantum dot hole qubit.
\newblock {\em Nature Physics}, 7:872, Aug 2011.
\newblock Article.

\bibitem{Emary2007}
C.~Emary, Xiaodong Xu, D.~G. Steel, S.~Saikin, and L.~J. Sham.
\newblock Fast initialization of the spin state of an electron in a quantum dot
  in the {Voigt} configuration.
\newblock {\em Phys. Rev. Lett.}, 98:047401, Jan 2007.

\bibitem{Imamoglu2009}
G.~Fernandez, T.~Volz, R.~Desbuquois, A.~Badolato, and
  A.~Imamo\ifmmode~\breve{g}\else \u{g}\fi{}lu.
\newblock Optically tunable spontaneous {Raman} fluorescence from a single
  self-assembled {InGaAs} quantum dot.
\newblock {\em Phys. Rev. Lett.}, 103:087406, Aug 2009.

\bibitem{Sun2016}
Z.~Sun, A.~Delteil, S.~Faelt, and A.~Imamo\ifmmode~\breve{g}\else \u{g}\fi{}lu.
\newblock Measurement of spin coherence using {Raman} scattering.
\newblock {\em Phys. Rev. B}, 93:241302, Jun 2016.

\bibitem{Delteil2015}
Aymeric Delteil, Zhe Sun, Wei-bo Gao, Emre Togan, Stefan Faelt, and
  A.~Imamo\ifmmode~\breve{g}\else \u{g}\fi{}lu.
\newblock Generation of heralded entanglement between distant hole spins.
\newblock {\em Nature Physics}, 12:218, Dec 2015.

\bibitem{Stockill2017}
R.~Stockill, M.~J. Stanley, L.~Huthmacher, E.~Clarke, M.~Hugues, A.~J. Miller,
  C.~Matthiesen, C.~Le~Gall, and M.~Atat\"ure.
\newblock Phase-tuned entangled state generation between distant spin qubits.
\newblock {\em Phys. Rev. Lett.}, 119:010503, Jul 2017.

\bibitem{Johansson2012}
J.R. Johansson, P.D. Nation, and Franco Nori.
\newblock {QuTiP}: {An} open-source {Python} framework for the dynamics of open
  quantum systems.
\newblock {\em Computer Physics Communications}, 183(8):1760 -- 1772, 2012.

\bibitem{Johansson2013}
J.R. Johansson, P.D. Nation, and Franco Nori.
\newblock {QuTiP} 2: {A Python} framework for the dynamics of open quantum
  systems.
\newblock {\em Computer Physics Communications}, 184(4):1234 -- 1240, 2013.

\bibitem{Note2}
In practice, this assumption limits the size of the LCs that can be produced in
  this approach to less than ten.

\bibitem{Note3}
Time-stamping these photons does not impose any experimental challenges, as
  detector setups with $\approx 30$ ps readily resolve these phases due to the
  precession time being of the order of $\sim $2 ns for an external field of
  100 mT.

\bibitem{Note4}
We assume $\eta $ is the probability of obtaining a detector click if a photon
  was produced by the QD, i.e.~it also includes any photon losses in the
  setup.).

\bibitem{Note5}
The numerical calculations were performed using the Overhauser
  ensemble-averaged matrix operations defined in Sec.~\ref {app:fidel} of the
  Appendix.

\bibitem{Popp2005}
M.~Popp, F.~Verstraete, M.~A. Mart\'{\i}n-Delgado, and J.~I. Cirac.
\newblock Localizable entanglement.
\newblock {\em Phys. Rev. A}, 71:042306, Apr 2005.

\bibitem{Morley2017}
Sam Morley-Short, Sara Bartolucci, Mercedes Gimeno-Segovia, Pete Shadbolt, Hugo
  Cable, and Terry Rudolph.
\newblock Physical-depth architectural requirements for generating universal
  photonic cluster states.
\newblock {\em Quantum Science and Technology}, 3(1):015005, 2018.

\bibitem{Prechtel2016}
Jonathan~H. Prechtel, Andreas~V. Kuhlmann, Julien Houel, Arne Ludwig, Sascha~R.
  Valentin, Andreas~D. Wieck, and Richard~J. Warburton.
\newblock Decoupling a hole spin qubit from the nuclear spins.
\newblock {\em Nat Mater}, 15(9):981--986, Sep 2016.
\newblock Article.

\bibitem{Note6}
The anisotropy in the hole g-factor is, in general, not the same as the
  effective anisotropy in the g-tensor for the hole Overhauser shift due to
  hlh--lh mixing.

\bibitem{Note7}
the averages are performed independently due to the statistical independence of
  $B^x_N$ and $\delta T$.

\bibitem{Smith2017}
Jake Iles-Smith, Dara P.~S. McCutcheon, Jesper M{\o}rk, and Ahsan Nazir.
\newblock Limits to coherent scattering and photon coalescence from solid-state
  quantum emitters.
\newblock {\em Phys. Rev. B}, 95:201305, May 2017.

\bibitem{Note8}
Admixture of conduction band states even in the absence of a lh contribution
  may result in a non-Ising type hyperfine Hamiltonian for the hh system \cite
  {Prechtel2016}, however, this goes beyond the scope of this work.

\bibitem{GershoniAtomistic}
M.~Zieli\ifmmode~\acute{n}\else \'{n}\fi{}ski, Y.~Don, and D.~Gershoni.
\newblock Atomistic theory of dark excitons in self-assembled quantum dots of
  reduced symmetry.
\newblock {\em Phys. Rev. B}, 91:085403, Feb 2015.

\bibitem{Dupertuis2011}
M.~A. Dupertuis, K.~F. Karlsson, D.~Y. Oberli, E.~Pelucchi, A.~Rudra, P.~O.
  Holtz, and E.~Kapon.
\newblock Symmetries and the polarized optical spectra of exciton complexes in
  quantum dots.
\newblock {\em Phys. Rev. Lett.}, 107:127403, Sep 2011.

\bibitem{GershoniPhenom}
Y.~{Don}, M.~{Zielinski}, and D.~{Gershoni}.
\newblock {The Optical Activity of the Dark Exciton}.
\newblock arXiv:1601.05530, January 2016.

\bibitem{Gross2006}
D.~Gross, K.~Kieling, and J.~Eisert.
\newblock Potential and limits to cluster-state quantum computing using
  probabilistic gates.
\newblock {\em Phys. Rev. A}, 74:042343, Oct 2006.

\bibitem{Segovia2015}
Mercedes Gimeno-Segovia, Pete Shadbolt, Dan~E. Browne, and Terry Rudolph.
\newblock From three-photon greenberger-horne-zeilinger states to ballistic
  universal quantum computation.
\newblock {\em Phys. Rev. Lett.}, 115:020502, Jul 2015.

\bibitem{Stinaff2006}
E.~A. Stinaff, M.~Scheibner, A.~S. Bracker, I.~V. Ponomarev, V.~L. Korenev,
  M.~E. Ware, M.~F. Doty, T.~L. Reinecke, and D.~Gammon.
\newblock Optical signatures of coupled quantum dots.
\newblock {\em Science}, 311(5761):636--639, 2006.

\bibitem{Krenner2005}
H.~J. Krenner, M.~Sabathil, E.~C. Clark, A.~Kress, D.~Schuh, M.~Bichler,
  G.~Abstreiter, and J.~J. Finley.
\newblock Direct observation of controlled coupling in an individual quantum
  dot molecule.
\newblock {\em Phys. Rev. Lett.}, 94:057402, Feb 2005.

\bibitem{Kim2011}
D.~Kim, S.~G. Carter, A.~Greilich, A.~S. Bracker, and Daniel Gammon.
\newblock Ultrafast optical control of entanglement between two quantum-dot
  spins.
\newblock {\em Nat Phys}, 7(3):223--229, Mar 2011.

\bibitem{Songmuang2003}
R.~Songmuang, S.~Kiravittaya, and O.~G. Schmidt.
\newblock Formation of lateral quantum dot molecules around self-assembled
  nanoholes.
\newblock {\em Applied Physics Letters}, 82(17):2892--2894, 2003.

\bibitem{Wang2008}
L.~Wang, A.~Rastelli, S.~Kiravittaya, P.~Atkinson, F.~Ding, C.~C. Bof~Bufon,
  C.~Hermannstaedter, M~Witzany, G.~J. Beirne, P.~Michler, and O.~G. Schmidt.
\newblock Towards deterministically controlled {InGaAs}/{GaAs} lateral quantum
  dot molecules.
\newblock {\em New Journal of Physics}, 10(4):045010, 2008.

\bibitem{Vora2015}
Patrick~M. Vora, Allan~S. Bracker, Samuel~G. Carter, Timothy~M. Sweeney, Mijin
  Kim, Chul~Soo Kim, Lily Yang, Peter~G. Brereton, Sophia~E. Economou, and
  Daniel Gammon.
\newblock Spin-cavity interactions between a quantum dot molecule and a
  photonic crystal cavity.
\newblock {\em Nature Communications}, 6:7665 EP, Jul 2015.

\bibitem{Shaji2008}
Nakul Shaji, C.~B. Simmons, Madhu Thalakulam, Levente~J. Klein, Hua Qin,
  H.~Luo, D.~E. Savage, M.~G. Lagally, A.~J. Rimberg, R.~Joynt, M.~Friesen,
  R.~H. Blick, S.~N. Coppersmith, and M.~A. Eriksson.
\newblock Spin blockade and lifetime-enhanced transport in a few-electron
  si/sige double quantum dot.
\newblock {\em Nature Physics}, 4:540, Jun 2008.

\bibitem{Greilich2011}
A.~Greilich, S.~G. Carter, Danny Kim, A.~S. Bracker, and D.~Gammon.
\newblock Optical control of one and two hole spins in interacting quantum
  dots.
\newblock {\em Nat Photon}, 5(11):702--708, Nov 2011.

\bibitem{Weiss2012}
K.~M. Weiss, J.~M. Elzerman, Y.~L. Delley, J.~Miguel-Sanchez, and
  A.~Imamo\ifmmode~\breve{g}\else \u{g}\fi{}lu.
\newblock Coherent two-electron spin qubits in an optically active pair of
  coupled {InGaAs} quantum dots.
\newblock {\em Phys. Rev. Lett.}, 109:107401, Sep 2012.

\bibitem{Elzerman2011}
J.~M. Elzerman, K.~M. Weiss, J.~Miguel-Sanchez, and
  A.~Imamo\ifmmode~\check{g}\else \v{g}\fi{}lu.
\newblock Optical amplification using raman transitions between spin-singlet
  and spin-triplet states of a pair of coupled {InGaAs} quantum dots.
\newblock {\em Phys. Rev. Lett.}, 107:017401, Jun 2011.

\bibitem{Vamivakas2010}
A.~N. Vamivakas, C.-Y. Lu, C.~Matthiesen, Y.~Zhao, S.~F{\"a}lt, A.~Badolato,
  and M.~Atat{\"u}re.
\newblock Observation of spin-dependent quantum jumps via quantum dot resonance
  fluorescence.
\newblock {\em Nature}, 467:297, Sep 2010.

\bibitem{Kim2008}
Danny Kim, Sophia~E. Economou, \ifmmode \mbox{\c{S}}\else \c{S}\fi{}tefan~C.
  B\ifmmode~\u{a}\else \u{a}\fi{}descu, Michael Scheibner, Allan~S. Bracker,
  Mark Bashkansky, Thomas~L. Reinecke, and Daniel Gammon.
\newblock Optical spin initialization and nondestructive measurement in a
  quantum dot molecule.
\newblock {\em Phys. Rev. Lett.}, 101:236804, Dec 2008.

\bibitem{Bolduc2017}
Eliot Bolduc, George~C. Knee, Erik~M. Gauger, and Jonathan Leach.
\newblock Projected gradient descent algorithms for quantum state tomography.
\newblock {\em NPJ Quantum Information}, 3(1):44, 2017.

\bibitem{Kwiat2001}
Daniel F.~V. James, Paul~G. Kwiat, William~J. Munro, and Andrew~G. White.
\newblock Measurement of qubits.
\newblock {\em Phys. Rev. A}, 64:052312, Oct 2001.

\bibitem{Paris2004}
M.~Paris and J.~Rehacek.
\newblock {\em Quantum State Estimation}.
\newblock Lecture Notes in Physics. Springer Berlin Heidelberg, 2004.

\bibitem{Ramsay2016}
A.~J. Ramsay.
\newblock Passive stabilization of a hole spin qubit using the optical {Stark}
  effect.
\newblock {\em Phys. Rev. B}, 93:075303, Feb 2016.

\end{thebibliography}
\bibliographystyle{unsrt}

\end{document}